\PassOptionsToPackage{unicode}{hyperref}
\PassOptionsToPackage{hyphens}{url}
\PassOptionsToPackage{dvipsnames,svgnames,x11names}{xcolor}
\documentclass[
  12pt]{article}

\usepackage{amsmath,amssymb,amsthm}
\usepackage{iftex}
\ifPDFTeX
  \usepackage[T1]{fontenc}
  \usepackage[utf8]{inputenc}
  \usepackage{textcomp} 
\else 
  \usepackage{unicode-math}
  \defaultfontfeatures{Scale=MatchLowercase}
  \defaultfontfeatures[\rmfamily]{Ligatures=TeX,Scale=1}
\fi
\usepackage{lmodern}
\ifPDFTeX\else  
\fi
\IfFileExists{upquote.sty}{\usepackage{upquote}}{}
\IfFileExists{microtype.sty}{
  \usepackage[]{microtype}
  \UseMicrotypeSet[protrusion]{basicmath} 
}{}
\makeatletter
\@ifundefined{KOMAClassName}{
  \IfFileExists{parskip.sty}{%
    \usepackage{parskip}
  }{
    \setlength{\parindent}{0pt}
    \setlength{\parskip}{6pt plus 2pt minus 1pt}}
}{
  \KOMAoptions{parskip=half}}
\makeatother
\usepackage{xcolor}
\makeatletter
\ifx\paragraph\undefined\else
  \let\oldparagraph\paragraph
  \renewcommand{\paragraph}{
    \@ifstar
      \xxxParagraphStar
      \xxxParagraphNoStar
  }
  \newcommand{\xxxParagraphStar}[1]{\oldparagraph*{#1}\mbox{}}
  \newcommand{\xxxParagraphNoStar}[1]{\oldparagraph{#1}\mbox{}}
\fi
\ifx\subparagraph\undefined\else
  \let\oldsubparagraph\subparagraph
  \renewcommand{\subparagraph}{
    \@ifstar
      \xxxSubParagraphStar
      \xxxSubParagraphNoStar
  }
  \newcommand{\xxxSubParagraphStar}[1]{\oldsubparagraph*{#1}\mbox{}}
  \newcommand{\xxxSubParagraphNoStar}[1]{\oldsubparagraph{#1}\mbox{}}
\fi
\makeatother

\providecommand{\tightlist}{%
  \setlength{\itemsep}{0pt}\setlength{\parskip}{0pt}}\usepackage{longtable,booktabs,array}
\usepackage{calc} 
\usepackage{etoolbox}
\makeatletter
\patchcmd\longtable{\par}{\if@noskipsec\mbox{}\fi\par}{}{}
\makeatother
\IfFileExists{footnotehyper.sty}{\usepackage{footnotehyper}}{\usepackage{footnote}}
\makesavenoteenv{longtable}
\usepackage{graphicx}
\makeatletter
\def\maxwidth{\ifdim\Gin@nat@width>\linewidth\linewidth\else\Gin@nat@width\fi}
\def\maxheight{\ifdim\Gin@nat@height>\textheight\textheight\else\Gin@nat@height\fi}
\makeatother
\setkeys{Gin}{width=\maxwidth,height=\maxheight,keepaspectratio}
\makeatletter
\def\fps@figure{htbp}
\makeatother

\makeatletter
\@ifpackageloaded{caption}{}{\usepackage{caption}}
\AtBeginDocument{%
\ifdefined\contentsname
  \renewcommand*\contentsname{Table of contents}
\else
  \newcommand\contentsname{Table of contents}
\fi
\ifdefined\listfigurename
  \renewcommand*\listfigurename{List of Figures}
\else
  \newcommand\listfigurename{List of Figures}
\fi
\ifdefined\listtablename
  \renewcommand*\listtablename{List of Tables}
\else
  \newcommand\listtablename{List of Tables}
\fi
\ifdefined\figurename
  \renewcommand*\figurename{Figure}
\else
  \newcommand\figurename{Figure}
\fi
\ifdefined\tablename
  \renewcommand*\tablename{Table}
\else
  \newcommand\tablename{Table}
\fi
}
\@ifpackageloaded{float}{}{\usepackage{float}}
\floatstyle{ruled}
\@ifundefined{c@chapter}{\newfloat{codelisting}{h}{lop}}{\newfloat{codelisting}{h}{lop}[chapter]}
\floatname{codelisting}{Listing}

\makeatother
\makeatletter
\@ifpackageloaded{caption}{}{\usepackage{caption}}
\@ifpackageloaded{subcaption}{}{\usepackage{subcaption}}
\makeatother

\ifLuaTeX
  \usepackage{selnolig}  
\fi
\usepackage[]{natbib}
\usepackage{bookmark}

\IfFileExists{xurl.sty}{\usepackage{xurl}}{} 
\hypersetup{
  pdftitle={Title},
  pdfauthor={Author 1; Author 2},
  pdfkeywords={3 to 6 keywords, that do not appear in the title},
  colorlinks=true,
  linkcolor={blue},
  filecolor={Maroon},
  citecolor={Blue},
  urlcolor={Blue},
  pdfcreator={LaTeX via pandoc}}

\usepackage{bbm}

\usepackage{multirow}
\usepackage{rotating}
\usepackage{graphicx}
\usepackage{tikz}
\usepackage{pgfplots}
\pgfplotsset{compat=1.18}
\usepackage{xfp} 

\usetikzlibrary {arrows.meta}

\newtheorem{theorem}{Theorem}
\newtheorem{proposition}{Proposition}
\newtheorem{corollary}{Corollary}
\newtheorem{lemma}{Lemma}

\usepackage[acronym]{glossaries}

\makeglossaries
\newacronym{ids}{\textit{IDS}}{\textit{intrusion detection system}}
\newacronym{it}{\textit{IT}}{\textit{information technology}}
\newacronym{roc}{\textit{ROC}}{\textit{receiver operating characteristic}}
\newacronym{cbi}{CBI}{\textit{conservative Bayesian inference}}
\newacronym{mdl}{\textit{MDL}}{\textit{Machine/Deep Learning}}
\newacronym{svm}{\textit{SVMs}}{\textit{support vector machines}}
\newacronym{ann}{\textit{ANNs}}{\textit{artificial neural networks}}
\newacronym{dnn}{\textit{DNNs}}{\textit{deep neural networks}}
\newacronym{pfd}{\textit{pfd}}{\textit{probability of failure on demand}}
\newacronym{iid}{i.i.d.}{\textit{independent identically distributed}}
\newacronym{pk}{\textit{PK}}{\textit{prior knowledge}}
\newacronym{mct}{MCT}{\textit{monotone convergence theorem}}
\newacronym{dct}{DCT}{\textit{dominated convergence theorem}}
\newacronym{cdf}{CDF}{\textit{cumulative distribution function}}
\newacronym{fdds}{\textit{fdds}}{\textit{finite-dimensional distributions}}
\newacronym{wrt}{\textit{w.r.t.}}{\textit{with respect to}}
\newacronym{wlog}{\textit{w.l.o.g.}}{\textit{without loss of generality}}
\newacronym{ivt}{IVT}{\textit{intermediate-value theorem}}
\newacronym{rhs}{\textit{r.h.s.}}{\textit{right hand side}}
\newacronym{lhs}{\textit{l.h.s.}}{\textit{left hand side}}
\newacronym{fp}{\textit{FP}}{false positive}
\newacronym{fn}{\textit{FN}}{false negative}
\newacronym{av}{\textit{AVs}}{autonomous vehicles}
\newacronym{pcm}{\textit{pcm}}{\textit{probability of a crash event per mile}}
\newacronym{pfm}{\textit{pfm}}{\textit{probability of a fatality event per mile}}
\newacronym{dpm}{dpm}{\textit{disengagements per mile}}
\newacronym{srgm}{\textit{SRGM}}{\textit{Software Reliability Growth Models}}
\newacronym{cots}{\textit{COTS}}{\textit{Commercial off-the-shelf}}

\newcommand{\anon}{1}

\begin{document}

\def\spacingset#1{\renewcommand{\baselinestretch}%
{#1}\small\normalsize} \spacingset{1}



\if1\anon
{
  \title{\bf Conservative Software Reliability Assessments Using Collections of Bayesian Inference Problems}
  \author{Kizito Salako\thanks{
    The authors acknowledge work by David Wright (communicated in an internal memo) that proved the necessary form of the prior \eqref{eqn_priorCBIsolBerNoFailures} using a different argument from that presented in this article. 
    }\hspace{.2cm}\\
    The Centre for Software Reliability, Department of Computer Science, \\City St. George's, University of London, \\Northampton square, EC1V 0HB, The United Kingdom\\
    and \\
    Rabiu Tsoho Muhammad \\
    The Centre for Software Reliability, Department of Computer Science, \\City St. George's, University of London, \\Northampton square, EC1V 0HB, The United Kingdom}
  \maketitle
} \fi

\if0\anon
{
  \bigskip
  \bigskip
  \bigskip
  \begin{center}
    {\LARGE\bf Conservative Software Reliability Assessments Using Collections of Bayesian Inference Problems}
\end{center}
  \medskip
} \fi

\bigskip
\begin{abstract}
When using Bayesian inference to support conservative software reliability assessments, it is useful to consider a collection of Bayesian inference problems, with the aim of determining the worst-case value (from this collection) for a posterior predictive probability that characterizes how reliable the software is. Using a Bernoulli process to model the occurrence of software failures, we explicitly determine (from collections of Bayesian inference problems) worst-case posterior predictive probabilities of the software operating without failure in the future. We deduce asymptotic properties of these conservative posterior probabilities and their priors, and illustrate how to use these results in assessments of safety-critical software. This work extends \emph{robust} \emph{Bayesian} \emph{inference} results and so-called \emph{conservative} \emph{Bayesian} \emph{inference} methods.
\end{abstract}

\noindent%
\vfill

\newpage
\spacingset{1.8} 

\section{Introduction}
Bayesian inference can support conservative assessments of software reliability. When assessing safety-critical software~---~whose primary function is ensuring safety (e.g., nuclear-power-plant safety protection systems or protection functions in autono\-mous vehicles)~---~an assessor will typically justify their beliefs (about how reliable the software is) with evidence. Evidence can include extensive successful operation of the software in other environments, use of code-blocks within the software that have undergone detailed analyses for correctness and validity, and the ``quality'' of the software development process (including the expertise of the development team); for various examples of evidence, see  \cite{LittlewoodWright1995SafeComp,LittlewoodWright2007TSE},  \cite{LittlewoodStrigini_1993},  \cite{NeilLittlewoodFenton1996SCSC},  \cite{BishopBloomfield2000},  \cite{Popov2013RESS}, \cite{Klein2009-seL4}. 
In principle, the assessor's beliefs are about statistical model parameters that characterize software reliability, such as the probability of the software failing~---~these parameters' true values are unknown, so the assessor expresses their uncertainty about the values (i.e. expresses their beliefs) as a prior distribution over possible parameter values. To be cautious, the assessor might only partly specify this prior, by specifying only probabilities of those parameter value ranges they have ``strong'' evidence for. In effect, they are specifying a \emph{credal set}~---~a \emph{(convex) collection of prior distributions that are consistent with these partial specifications}. 

Upon observing the failure behavior of the software during software testing, the assessor updates their beliefs and makes conservative claims about whether the software is sufficiently reliable; claims supported by posterior probabilities of interest. These posterior probabilities form part of the justification for certifying, or rejecting, the software: examples include the assessor's confidence in an upper (credible) bound on the probability of the software failing at random, or the (predictive) probability that the software will not fail during future executions. In all such cases, the assessor is seeking the most undesirable value for the posterior probability of interest --- the assessor is being as critical as they can be of the software, given the software's observed performance and the available evidence. What \emph{is} this worst-case posterior probability value, and \emph{which} prior distribution (consistent with the assessor's beliefs) gives this value? 

This question can be formalized as a constrained optimization over a collection of prior distributions. Let $(\Omega,\Sigma)$ be a measurable space and $X:\Omega\rightarrow I$ be a random variable mapping onto the compact interval $I\subset\mathbb R$. Let $\mathcal D$ be the set of all probability measures defined on $\Sigma$. Consider a pair of continuous, real-valued, non-negative functions $f,\,g$, with shared domain $I$. Since $I$ is compact, $f(X)$ and $g(X)$ are bounded, so $f(X),\, g(X)\in \mathbb L^1(\Omega,\,\mathbb P)$ for all $\mathbb P\in\mathcal D$. Partition $I$ into contiguous sub-intervals $I_1$, $\ldots$, $I_n$. The preimages $A_i:=X^{-1}(I_i)\in\Sigma$ partition $\Omega$. Assume $\inf_{I_j}g>0$ for some interval $I_j$. Let $0<p_1,\ldots,p_n<1$ satisfy $\sum_{i=1}^n p_i=1$. Our assessor wants the answer to the following optimization problem.
\begin{align}
&\underset{\mathcal D}{\inf}\frac{\mathbb E [f(X)]}{\mathbb E [g(X)]} \nonumber\\
\text{\emph{s.t.}}\hspace{0.25cm}&\mathbb P(A_1)=p_1,\ldots,\mathbb P(A_n)=p_n
\label{eqn_genCBIprob}
\end{align}
Problem~\eqref{eqn_genCBIprob} is solved when one can state the infimum and a prior distribution of $X$ that achieves it. Typically, for parameter value $x$, $g(x)$ is the likelihood of the software's observed failure behavior during operation, while $f(x)/g(x)$ is a probability of a \emph{yet-to-be-observed} (possibly \emph{unobservable}) event related to the software's reliability. Prior distributions in $\mathcal D$ that satisfy \eqref{eqn_genCBIprob}'s constraints are ``feasible'' and form a credal set. 

Problem~\eqref{eqn_genCBIprob} is not new. Generalizations have been extensively studied in the robust Bayesian inference literature (see  \cite{berger1990_SensitivityToPrior},  \cite{moreno1991robust},  \cite{1991_Lavine},  \cite{Berger_1994_RobustnessInBidinesionalModels}), and certain nonlinear fractional program forms of~\eqref{eqn_genCBIprob} can be solved using Dinkelbach-type fixed-point iteration (see  \cite{Dinkelbach1967},  \cite{Schaible1976Dinkelbach}),  while several special cases related to~\eqref{eqn_genCBIprob}~---~so-called \emph{conservative Bayesian inference} (\acrshort*{cbi}) problems~---~have been solved to support conservatism in reliability assessments (see  \cite{bishop_toward_2011},  \cite{strigini_software_2013},  \cite{zhao_assessing_2019,littlewood_reliability_2020},  \cite{salako_conservative_2021},  \cite{SalakoZhao_TSE_2023,kizito_QRE_2024}). Although \acrshort*{cbi} problems can often be solved numerically, explicit solutions tend to be only for special cases, and the solution methods employed can appear distinct and \emph{ad hoc}. This paper: \textbf{i)} unifies these solution methods, exemplified by an explicit general solution to a problem that arises often in software reliability assessments; \textbf{ii)} analyses this solution's properties; \textbf{iii)} highlights the practical implications of these results. Worst-case software reliability estimates are stated in Section~\ref{sec_results}, while Section~\ref{sec_discussion} discusses their properties, applications, limitations, and future work.

\section{Results}
\label{sec_results}
Proposition~\ref{prop_cbi_transform} simplifies problem~\eqref{eqn_genCBIprob} (\emph{cf.}  Theorem~1 of  \cite{moreno1991robust}, Theorem 3 of  \cite{BetroRuggeriMeczarski1994-JSPI}).
\begin{proposition}
\label{prop_cbi_transform}
For any $\mathbb P\in\mathcal D$, $
        \dfrac{\mathbb E[f(X)]}{\mathbb E[g(X)]}=\frac{\sum\limits_{i=1}^{n}f(x_i)\mathbb P(A_i)}{\sum\limits_{i=1}^{n}g(x_i)\mathbb P(A_i)}$  for some $x_i\in I_i$ ($i=1...n$).
\end{proposition}
\begin{proof} See Appendix~A, Supplementary Material. \end{proof}
 
By proposition \ref{prop_cbi_transform}, problem \eqref{eqn_genCBIprob} is equivalent to the nonlinear fractional program
\begin{align}
&\underset{x_1\in I_1,...,x_n\in I_n}{\inf}\frac{\sum_{i=1}^{n}f(x_i)p_i}{\sum_{i=1}^{n}g(x_i)p_i}\,. 
\label{eqn_morenoetalCBIprob}
\end{align} 
To illustrate solving \eqref{eqn_morenoetalCBIprob}, consider the following software reliability assessment scenario. An assessor is interested in the probability that a piece of on-demand software will operate correctly on the next $m$ demands, after observing the software correctly handle $k$ demands while also failing on $r$ demands. A Bernoulli (\acrshort*{iid}) process is judged 
to adequately model the software's successes and failures on a random sequence of demands during its operation. The Bernoulli process parameter $X$ is the unknown \emph{probability of failure on the next demand} (\emph{pfd}), and $\mathcal D$ is the set of all prior distributions over the unit interval (i.e. all possible prior distributions of $X$). Specifically, for $0=y_0<y_1<\ldots<y_n=1$, the assessor expresses prior beliefs about how likely $X$ is to lie in each of $n$ intervals $I_1:=[y_0,y_1]$, $I_2:=(y_1,y_2]$, \ldots, $I_n:=(y_{n-1},y_n]$, by singling out those $\mathbb P\in\mathcal D$ that obey $\mathbb P (X^{-1}(I_i))=p_{i}$ for $i=1,\ldots,n$. With $f(x)=x^r(1-x)^{m+k}$ and $g(x)=x^r(1-x)^{k}$, problem~\eqref{eqn_morenoetalCBIprob} becomes\begin{align}
\underset{x_1\in I_1,\ldots,x_n\in I_n}{\inf}\frac{\sum_{i=1}^{n}x_i^r(1-x_i)^{m+k} p_i}{\sum_{i=1}^{n}x_i^r(1-x_i)^{k} p_i}\,.
\label{eqn_morenoetalBer}
\end{align}
\begin{theorem}
\label{thm_gen_sol}
A unique fixed point solves \eqref{eqn_morenoetalBer}~---~i.e. a triplet, $\phi^\ast$, $y_\ast$, $y_{**}$, that satisfies
\begin{gather*}
(1-y_{**})^m\left(\frac{r-y_{**}(m+k+r)}{r-y_{**}(k+r)}\right)=\phi^\ast = (1-y_\ast)^m\left(\frac{r-y_\ast(m+k+r)}{r-y_\ast(k+r)}\right),\\ \phi^\ast=\min\{\phi^\ast_1,\phi^\ast_2\}\,,
\end{gather*}where
\begin{align*}
\phi^\ast_1&=\left\{\begin{array}{ll}\dfrac{\sum_{i=2}^{j_1}f(y_{i-1})p_i+\sum_{i=j_1+1}^{j_2-1}f(y_{i})p_i+f({y_\ast})p_{j_2}+\sum_{i=j_2+1}^{n}f(y_{i-1})p_i}{\sum_{i=2}^{j_1}g(y_{i-1})p_i+\sum_{i=j_1+1}^{j_2-1}g(y_{i})p_i+g({y_\ast})p_{j_2}+\sum_{i=j_2+1}^{n}g(y_{i-1})p_i}\,,\hfill j_1<j_2\\
\dfrac{\sum_{i=2}^{n}f(y_{i-1})p_i}{\sum_{i=2}^{n}g(y_{i-1})p_i}\,,\hfill j_1=j_2\end{array}\right.\\
\phi^\ast_2&=\left\{\begin{array}{ll}\dfrac{\sum_{i=2}^{j_1-1}f(y_{i-1})p_i+\sum_{i=j_1}^{j_2-1}f(y_{i})p_i+f({y_\ast})p_{j_2}+\sum_{i=j_2+1}^{n}f(y_{i-1})p_i}{\sum_{i=2}^{j_1-1}g(y_{i-1})p_i+\sum_{i=j_1}^{j_2-1}g(y_{i})p_i+g({y_\ast})p_{j_2}+\sum_{i=j_2+1}^{n}g(y_{i-1})p_i}\,,\hfill j_1<j_2\\
\dfrac{\sum_{i=2}^{j_2-1}f(y_{i-1})p_i+f(y_\ast)p_{j_2}+\sum_{i=j_2+1}^{n}f(y_{i-1})p_i}{\sum_{i=2}^{j_2-1}g(y_{i-1})p_i+g(y_\ast)p_{j_2}+\sum_{i=j_2+1}^{n}g(y_{i-1})p_i}\,,\hfill j_1=j_2\end{array}\right.
\end{align*}for some $1\leqslant j_1\leqslant j_2\leqslant n$ such that $\ y_{**}<\frac{r}{r+m+k}<\frac{r}{r+k}<y_\ast$, $\ y_{**}\in \bar{I}_{j_1}$, and $\ y_\ast\in \bar{I}_{j_2}$ (n.b. $\bar{I}_i$ is the closure of $I_{i}$). The infimum $\phi^*$ is attained by the prior distribution of $X$
\begin{align}
\label{eqn_priorCBIsolBer}
    \mathbb P(X=x)=\left\{\begin{array}{cl}
    p_1,  & \text{if } x=\mathbbm{1}_{n=1},\\ 
     p_{i}, & \left\{\begin{array}{l}\text{if } x=y_{i-1} \text{ and } \{2\leqslant i<j_1\}\cup\{j_2<i\leqslant n\},\\\text{if } x=y_{i} \text{ and } \{j_{1}<i<j_{2}\},\end{array}\right.\\
        p_{j_1}, & \text{if } x=y_{j_{1}-1}\mathbbm{1}_{{\phi}^*=\phi_1^*}+\min\{y_{j_1},\,y_*\}\mathbbm{1}_{{\phi^*}=\phi_2^*\ne\phi_1^*},\\
           p_{j_2}, & \text{if } x=\left(y_{j_{1}-1}\mathbbm{1}_{j_1=j_2}+y_{*}\mathbbm{1}_{j_1<j_2}\right)\mathbbm{1}_{{\phi}^*=\phi_1^*}+y_{*}\mathbbm{1}_{{\phi^*}=\phi_2^*\ne\phi_1^*},\\
           0, & \text{otherwise.}
    \end{array}\right.
\end{align}
\end{theorem}
\begin{proof} See Appendix~B, Supplementary Material.
\end{proof}
\begin{corollary}\label{cor_special_case_no_failure_solution}
Consider the limiting case of \eqref{eqn_morenoetalBer} with no software failures ($r=0$),
\begin{align}
\label{eq_obj_r=0}
\underset{x_1\in I_1,\ldots,x_n\in I_n}{\inf}\frac{\sum_{i=1}^{n}(1-x_i)^{m+k} p_i}{\sum_{i=1}^{n}(1-x_i)^{k} p_i}
\end{align}
A unique fixed point solves \eqref{eq_obj_r=0}~---~i.e. a pair, $\phi^\ast$, $y_\ast$, that satisfies 
\begin{align*}
\phi^*&=(1-y_*)^{m}(m+k)/k\,,\\
\phi^*&=\frac{\sum_{i=1}^{j-1}(1-y_{i})^{m+k}p_{i}+(1-y_*)^{m+k}p_{j}+\sum_{i=j+1}^{n}(1-y_{i-1})^{m+k}p_{i}}{\sum_{i=1}^{j-1}(1-y_{i})^{k}p_{i}+(1-y_*)^{k}p_{j}+\sum_{i=j+1}^{n}(1-y_{i-1})^{k}p_{i}}\,, 
\end{align*}
for some $1\leqslant j\leqslant n$ with $y_*\in \bar{I}_j$. The infimum $\phi^*$ is attained by the prior distribution of $X$ 
\begin{align}
    \mathbb P(X=x)=\left\{\begin{array}{cl}
     p_{i},&\left\{\begin{array}{l}\text{if } x=y_i \text{ and } i<j,\\\text{if } x=y_{i-1} \text{ and } i>j,\end{array}\right.\\
        p_{j}, & \text{if } x=y_*, \\
           0, & \text{otherwise.}
    \end{array}\right.
\label{eqn_priorCBIsolBerNoFailures}
\end{align}
\end{corollary}
\begin{proof} Let $r\to0$ in Theorem~\ref{thm_gen_sol}, assuming $n+1$ constraining intervals where the first interval's upper endpoint is $r/(r+k)$ and the first interval's probability is $o(1)$.\hfill\qedhere
\end{proof}
%
%
\begin{figure}[h!]
\centering
\scalebox{0.9}{
   \begin{tikzpicture}
  \draw[->] (0,0) -- (14,0);
\draw[darkgray, line width=2.5pt] (0,-0.005) -- (0,0.90) node[below=5pt] {\space};
\draw[darkgray, line width=2.5pt] (6.5,-0.005) -- (6.5,0.5) node[below=5pt] {\space};
\draw[darkgray, line width=2.5pt] (2,-0.005) -- (2,0.75) node[below=5pt] {\space};
\draw[darkgray, line width=2.5pt] (3.5,-0.005) -- (3.5,0.9) node[below=5pt] {\space};
\draw[darkgray, line width=2.5pt] (12.5,-0.005) -- (12.5,0.3) node[below=5pt] {\space};
\draw[darkgray, line width=2.5pt] (8.7,-0.005) -- (8.7,0.61) node[below=5pt] {\space};
\draw[darkgray, line width=2.5pt] (9.7,-0.005) -- (9.7,0.99) node[below=5pt] {\space};
\draw[darkgray, line width=2.5pt] (10.5,-0.005) -- (10.5,0.74) node[below=5pt] {\space};
   \draw (0,0) -- (0,-0.15) node[below=5pt] at (0,0) {$0$};
  \foreach \x/\label in {
    1/{\dots},
    2/{$y_{j_{1}-2}$},
    3.5/{$y_{{j_1}-1}$},
    5.5/{$y_{j_1}$},
    6.5/{$y_{j_1+1}$},
    7.5/{\dots},
    8.7/{$y_{{j_2}-1}$},
    10.5/{$y_{j_2}$},
    11.5/{\dots},
    12.5/{$y_{n-1}$},
    13.5/{$1$}
  }{
    \draw (\x,0) -- (\x,-0.15);
    \node[below=5pt] at (\x,0) {\label};
  }
  \node[below=5.4pt] at (4.6,0) {$y_{**}$};
    \node[below=2.5pt] at (8,0) {\space};
  \node[below=5.4pt] at (9.7,0) {$y_{*}$};
\node[above=22pt] at (0,0.1) {$p_1$};
\node[above=12pt] at (6.5,0.1) {$p_{j_1+1}$};
\node[above=15pt] at (2,0.1) {$p_{{j_1}-1}$};
 \node[above=22pt] at (3.5,0.1) {$p_{j_1}$};
  \node[above=12pt] at (8.7,0.1) {$p_{{j_2}-1}$};
  \node[above=23pt] at (9.7,0.1) {$p_{j_2}$};
   \node[above=15pt] at (10.5,0.1) {$p_{{j_2}+1}$}; \node[above=5pt] at (12.5,0.1) {$p_{n}$};
\end{tikzpicture}}
\caption{\textit{An illustration of conservative prior distribution \eqref{eqn_priorCBIsolBer}}}
\label{fig_worst_prior}
\end{figure}
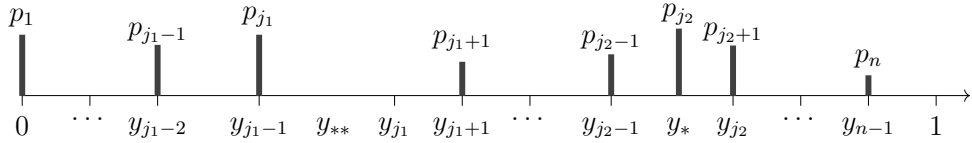

\section{Discussion}
\label{sec_discussion}
\subsection{Bayesian Inference for Conservative Assessments}
Theorem~\ref{thm_gen_sol} and Corollary~\ref{cor_special_case_no_failure_solution} illustrate a conservative application of Bayesian inference. These results determine the smallest possible posterior predictive probability of the software successfully handling $m$ demands in the future, after observing the software's past performance and accounting for any evidence that justifies only a partial specification of prior beliefs about the \acrshort*{pfd}. 
The assessor's partially specified beliefs define a collection of prior distributions that each satisfy this specification. So, each prior is consistent with the available reliability evidence. These ``feasible'' priors are updated via traditional Bayesian inference, producing a collection of posterior distributions and a ``worst-case'' value (from this collection) for a posterior probability of interest. Using inference in this way to support conservative assessments is so-called \acrshort*{cbi}.


The conservative prior distributions \eqref{eqn_priorCBIsolBer} and \eqref{eqn_priorCBIsolBerNoFailures} are ``degenerate'': i.e. discrete priors for a parameter that can be any value in the unit interval. Most values of the parameter~---~of the \acrshort*{pfd}~---~are impossible with these priors. These priors do \emph{not} represent the assessor's beliefs \emph{per se}, but represent beliefs (consistent with the assessor's) that conservatively interpret the evidence for inference. The assessor's actual beliefs are represented by the constraints on the collection of feasible priors.

The posterior predictive probabilities under \acrshort*{cbi} ``converge'' to the values obtained under traditional Bayesian inference, as the partial specification of the prior approaches that of a fully specified prior (\emph{cf.} convergence of $\epsilon$-contaminations to traditional Bayesian inference, in \cite{berger1990_SensitivityToPrior,moreno1991robust} for instance, as $\epsilon\to 0$). For example, suppose the constraints in \eqref{eqn_morenoetalBer} are a partial specification of a particular Beta distribution. Further suppose a sequence of increasingly more detailed partial specifications of this Beta distribution, with the number of ``$I_i$'' intervals growing to infinity while the intervals' maximum size shrinks to zero. Then, by the \emph{monotone} \emph{convergence} \emph{theorem}, both the numerator and denominator of the objective function in \eqref{eqn_morenoetalBer} converge to definite integrals w.r.t. the Beta distribution. So, the infimum in \eqref{eqn_morenoetalBer} converges to the posterior predictive probability (of the software surviving $m$ demands without failure) resulting from using the Beta prior distribution. 

The partial specification does not need to be in terms of parameter intervals (see generalizations as measurable sets in  \cite{moreno1991robust}). Some alternative specifications include perturbations of priors (so-called $\epsilon$-contaminations) and moments; e.g. see  \cite{berger1994_robustBayesianOverview},  \cite{BetroRuggeriMeczarski1994-JSPI},\cite{Salako_QEST_2020}. In  \cite{SalakoZhao_TSE_2023}, a partial specification of a joint prior distribution states that one of its marginal distributions is continuous (without specifying a distribution family). 

Strictly speaking, the conservative prior distribution may not be feasible (i.e. it does not satisfy all of the problem constraints). Rather, it is the limit of feasible priors. For example, for the conservative prior \eqref{eqn_priorCBIsolBer}, probability $p_i$ might be assigned to the lower endpoint of interval $I_i$ despite $I_i$ being open at its lower endpoint.


When applying these results, Theorem~\ref{thm_gen_sol} and Corollary~\ref{cor_special_case_no_failure_solution} are useful in two ways. Firstly, they improve numerical solution-finding routines, by focusing the routines to be more targeted in their search for solutions. For instance, limiting the search to those functional forms that satisfy the fixed-point system of equations in Theorem\ \ref{thm_gen_sol}; the optimization becomes a root-finding problem. Secondly, the solutions' properties give insight into the practical meaning and implications of these solutions; e.g. insight into an assessor's conservatism and the claims the assessor can justify.

These results reemphasize a general observation made in previous works: that \acrshort*{cbi} conservatism is relative, not absolute. The conservative values for different posterior predictive probabilities are achieved using different priors. A prior is conservative, or not, depending on what is observed, on which beliefs are justified, on what statistical model is used, and on what claim is being made. For instance, the conservative priors \eqref{eqn_priorCBIsolBer} and \eqref{eqn_priorCBIsolBerNoFailures} differ from those in  \cite{zhao_assessing_2019}, because the posterior predictive probabilities differ: this paper is concerned with the posterior probability of the software successfully handling future demands, while that paper considers posterior credible bounds on \emph{pfd} $X$. For more example posterior predictive probabilities and their conservative priors, 
see  \cite{zhao_conservative_2015},  \cite{kizito_QRE_2024}.

Conservatism, rather than robustness, is this paper's focus~---~the difference between the two is worth highlighting. \emph{Robust Bayesian inference} techniques evaluate the sensitivity of the outputs of inference to changes in, or ignorance about, the inputs to inference~---~inputs such as the prior distribution and the likelihood (see \cite{1991_Lavine}, \cite{berger1994_robustBayesianOverview}). These techniques also utilise collections of Bayesian inference problems, to determine upper and lower bounds on posterior probabilities of interest. The gap between these bounds quantifies how robust the inference is, so that a sufficiently small gap could mean a certain level of ignorance about the prior or likelihood is acceptable. There are also approaches to robust inference using \emph{imprecise probabilities}, involving convex collections of prior distributions and ranges of values for probabilities (see \cite{walley1991statistical}, \cite{Augustin2014}). Contrastingly, conservatism, as presented in this paper, attempts to ``\emph{err-on-the-side-of-caution}'' by determining if the worst-possible value for a posterior probability is, nevertheless, acceptable. This can be viewed as a special case of \emph{conditional $\Gamma$-maximin decision making} with two actions~---~``accept'' and ``reject''~---~ and two corresponding ``utility'' functions~---~respectively, the unknown probability of the software correctly handling $m$ demands (i.e. $(1-X)^m$) and some acceptable smallest value for this probability (see \cite{BetroRuggeri1992-ConditionalGammaMinimax} and \cite{Vidakovic2000-GammaMinimax} for general conditional $\Gamma$-minimax treatments in terms of losses rather than utilities). All of these related, philosophically different, approaches employ similar techniques to address model misspecification, belief elicitation, and uncertainty-aware decision-making challenges.

\acrshort*{cbi} can be used to check the robustness of \acrshort*{cbi} results that rely on an assumed statistical model. For example, problem\ \eqref{eqn_morenoetalBer} assumes the software has \acrshort*{iid} outcomes on randomly occurring demands. However, in many situations, there are good reasons to expect software failures and successes to be correlated. How sensitive are \acrshort*{cbi}'s conservative results to deviations from the \acrshort*{iid} model? By using a more general statistical model that has the \acrshort*{iid} model as a special case (i.e. Klotz's model in \cite{klotz_statistical_1973}),  \acrshort*{cbi} is used in  \cite{SalakoZhao_TSE_2023,kizito_QRE_2024} to show that small deviations can have a significant impact.

\subsection{Properties of CBI Solutions}
The infimum of the objective function in \eqref{eqn_morenoetalBer} can only take on certain functional forms determined by the gradient of the objective function. There are $n$ preferred locations in the prior's domain where the $p_i$ probabilities should be assigned; the $i$-th location is a limit point of constraint interval $I_i$. All previously published \acrshort*{cbi} solutions have each been constrained by finitely many such functional forms. 

Multiparameter generalizations of \eqref{eqn_morenoetalCBIprob} have analogous conservative priors to \eqref{eqn_priorCBIsolBer} and \eqref{eqn_priorCBIsolBerNoFailures}: the priors remain discrete but become multidimensional, they agree with the univariate conservative priors in appropriate limits, and they also result from unique fixed points.  \cite{littlewood_reliability_2020} contains some examples. Other multiparameter problems, with objective functions that are posterior credible bounds on \acrshort*{pfd}, can be found in  \cite{zhao_assessing_2020},  \cite{salako_conservative_2021},  \cite{kizito_QRE_2024}. The conservative priors remain discrete, multidimensional, but have no explicit fixed-point in the \acrshort*{cbi} solution statement.

Whether the fixed point dependence is implicit or explicit in the solution, a common theme across these solutions is that the fixed points behave like ``attractors'' or ``repellers''~---~attracting or repelling locations where probability masses are assigned in the domain of the prior. The optimization constraints induce a partition of the prior's domain, and the fixed points dictate those locations where probability mass should be assigned (in each subset of the partition) to achieve conservative results. For some \acrshort*{cbi} problems, the fixed point exclusively attracts or repels. For other problems, the fixed point attracts locations in some subsets while repelling locations in other subsets. An example is the five solutions in the proof of the theorem in  \cite{zhao_assessing_2019}; the placement of probability masses for each of these solutions is determined by the attraction and repulsion of the fixed point $r/(r+k)$ (i.e. the MLE for $x^r(1-x)^{k}$). 

Theorem~\ref{thm_gen_sol}'s fixed-point triplet~---~i.e. $\phi^*$, $y_{**}$, $y_*$~---~are related by the function $h: [0,1]\setminus\{\frac{r}{r+k}\}\rightarrow \mathbb R$, $h(x) = (1 - x)^m\left(\frac{r - x(m + k + r)}{r - x(k + r)}\right)$; see Figure~\ref{fig_h_stationarypoints}. The gradient of the objective function defines $h$, where $h$ determines when gradient components are non-trivially zero. The roots of the difference $h-\phi$, between $h$ and the objective function $\phi$, imply the fixed point ``attracts'' (w.r.t. $y_{*}$) and ``repels'' (w.r.t. $y_{**}$). Similarly, for Corollary \ref{cor_special_case_no_failure_solution}, $y_{*}$ ``attracts'' (see Appendix~B, Supplementary Material).

The ``attractors'' and ``repellers'' are not necessarily ``points''; they can be topological manifolds in the prior's domain. In \cite{SalakoZhao_TSE_2023}, a joint prior is constrained to have a continuous marginal. Consequently, the conservative joint prior has line-segments in its domain that ``attract'' univariate probability density. Similar degenerate solutions occur for multidimensional $\epsilon$-contamination problems; e.g. solutions involving 2-dimensional extreme priors with support on 1-dimensional curves in the domain of the priors~---~ see Theorem 1 and Remark 1 of \cite{LiseoMorenoSalinetti1996-GivenMarginals}, Remark 2 of \cite{MorenoCano1995-CollectionOfSets}, and \cite{LavineWassermanWolpert1991-SpecifiedMarginals}.

\def\rval{1}
\def\mval{3}
\def\kval{2}
\def\xsing{\fpeval{\rval / (\rval + \kval)}}
\def\xroot{\fpeval{\rval / (\rval + \mval + \kval)}}
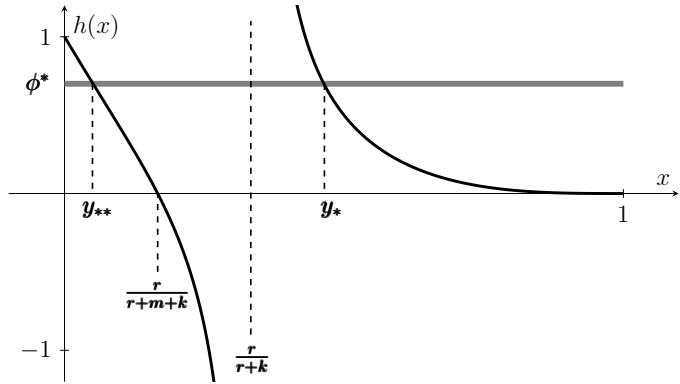
\begin{figure}[htbp!]
\begin{center}
\scalebox{0.78}
{\begin{tikzpicture}
  \begin{axis}[
    width=13cm,
    height=8cm,
    domain=0:1,
    samples=400,
    xlabel={$x$},
    ylabel={$h(x)$},
    axis lines=middle,
    xtick distance=1,
    ytick distance=1,
    enlargelimits=true,
    tick style={black},
    ymin=-1, ymax=1,
  ]

  \draw[line width=1mm,gray] ({axis cs:0,0.7}) -- ({axis cs:1,0.7});
  \node[black] at (axis cs:-0.01,0.7) [anchor=east] {$\pmb{\phi^*}$};

  \addplot[
    domain=0:\fpeval{\xsing - 0.001},
    line width = 0.5mm,
    black,
  ]
  {(1 - x)^\mval * (\rval - x*(\mval + \kval + \rval)) / ( \rval - x*(\kval + \rval) )};

  \addplot[
    domain=\fpeval{\xsing + 0.001}:0.999,
    line width = 0.5mm,
    black,
  ]
  {(1 - x)^\mval * (\rval - x*(\mval + \kval + \rval)) / ( \rval - x*(\kval + \rval) )};


  \draw[dashed,thick,black] ({axis cs:\xsing,-0.9}) -- ({axis cs:\xsing,1.1});
  \node[black] at (axis cs:\xsing+0.003,-1.06) {$\pmb{\frac{r}{r+k}}$};


  \draw[dashed,thick,black] ({axis cs:\xroot,-0.5}) -- ({axis cs:\xroot,0});
  \node[black] at (axis cs:\xroot,-0.65) {$\pmb{\frac{r}{r+m+k}}$};

  \draw[dashed,thick,black] ({axis cs:0.05,0.7}) -- ({axis cs:0.05,0});
  \node[black] at (axis cs:0.06,-0.1) {$\pmb{y_{**}}$};

  \draw[dashed,thick,black] ({axis cs:0.465,0.7}) -- ({axis cs:0.465,0});
  \node[black] at (axis cs:0.479,-0.1) {$\pmb{y_*}$};

  \end{axis}
\end{tikzpicture}}
\end{center}\caption{\textit{How $\phi^*$, $y_{**}$, and $y_*$, are related by $h(x) = (1 - x)^m\left(\frac{r - x(m + k + r)}{r - x(k + r)}\right)$.}} 
\label{fig_h_stationarypoints}
\end{figure}

Problem\ \eqref{eqn_morenoetalBer} is appropriate for safety-critical applications, where \emph{any} software failure is unacceptable. For less critical applications, future failures might be acceptable or, even, inevitable. A generalisation of \eqref{eqn_morenoetalBer} considers, at most, $l$ future failures:  
\begin{align}
\label{eqn_obj_cappedfuturefailures}
\underset{x_1\in I_1,\ldots,x_n\in I_n}{\inf}\frac{\sum_{i=1}^{n}\sum_{s=0}^l{m\choose s}x_i^{r+s}(1-x_i)^{m+k-s} p_i}{\sum_{i=1}^{n}x_i^r(1-x_i)^{k} p_i}\,,
\end{align}
where $l<m$. The solution to \eqref{eqn_obj_cappedfuturefailures} is also governed by a fixed point with an attractor and repeller, like Theorem\ \ref{thm_gen_sol}. We focus on the simpler problem \eqref{eqn_morenoetalBer} to ease the exposition, and because \eqref{eqn_morenoetalBer} reveals the salient solution characteristics for this class of \acrshort*{cbi} problems.

\subsection{Practical Applications and Implications}
The infimum in Theorem~\ref{thm_gen_sol}, i.e. $\phi^*$, behaves as expected given the practical context. When positive operational evidence (of the software being very reliable) mounts~---~i.e. as the number, $k$, of observed successfully handled demands increases~---~$\phi^*$ increases monotonically. While, with more negative operational evidence~---~so, as the number, $r$, of software failures increases~---~$\phi^*$ monotonically reduces. Note:  $0\leqslant \phi^*\leqslant (1-y_2)^m$, where $0$ and $(1-y_2)^m$ are the limits of $\phi^*$ for, respectively, large $r$ and $k$ (see Appendix~C, Supplementary Material).


In practice, when failures and successes are observed (i.e. $r>0,\ k>0$), the constraints in problem~\eqref{eqn_morenoetalBer} \emph{must} have three (or more) intervals; otherwise $\phi^*=0$, no matter how many successes $k$ are observed! This is because $x_1$ must always be $0$ (see prior~\eqref{eqn_priorCBIsolBer}) and, with two intervals, $\phi=(1-x_2)^m$ so that $\phi\to 0$ as $x_2\to 1$. Hence, setting $x_2=1$ ensures $\phi^*=0$. One's initial reaction upon seeing $\phi^*=0$ might be that this is ``too'' conservative. Surely there must be \emph{some} non-zero probability of the software correctly handling the future demands~---~particularly when it has already correctly handled most of the demands in the past? This reaction is misguided; it ignores the fact that \acrshort*{cbi} can only be as conservative as the assessor's limited beliefs and observations will allow, so it highlights what unfavorable claims can be justified if the evidence is not strong enough to invalidate the claim. To preclude this $\phi^*=0$ possibility, the assessor must have evidence that strongly suggests the software is very reliable (so $y_2\ll1$) but not perfect (so $y_1>0$). In the rest of our discussion, applications of Theorem~\ref{thm_gen_sol} will involve constraints with three or four intervals. 
\begin{sidewaystable}[htbp!]
    \centering
    \resizebox{0.99\textwidth}{!}{\begin{tabular}{ccccccccccccc}
        \hline
        \multirow{4}{*}{$\begin{array}{c}\textbf{Number}\\\textbf{of failures} \\(r)\end{array}$} 
        & \multicolumn{4}{c}{\(\boldsymbol{m = 46,\ \alpha = 0.009895}\)} 
        & \multicolumn{4}{c}{\(\boldsymbol{m = 500,\ \alpha = 0.097982}\)} 
        & \multicolumn{4}{c}{\(\boldsymbol{m = 1000,\ \alpha = 0.178476}\)} \\
        \cline{2-13}
        & $\begin{array}{c}\textbf{Beta}\\\textbf{Prior}\end{array}$  
        & \multicolumn{3}{c}{\textbf{CBI Priors}} 
        &  $\begin{array}{c}\textbf{Beta}\\\textbf{Prior}\end{array}$
        & \multicolumn{3}{c}{\textbf{CBI Priors}} 
        &  $\begin{array}{c}\textbf{Beta}\\\textbf{Prior}\end{array}$
        & \multicolumn{3}{c}{\textbf{CBI Priors}} \\
        \cline{2-5} \cline{6-9} \cline{10-13}
        &  
        & \(\boldsymbol{y_2=10^{-4}}\) & \(\boldsymbol{y_2=2 \times 10^{-5}}\) & \(\boldsymbol{y_2=10^{-5}}\) 
        &  
        & \(\boldsymbol{y_2=10^{-4}}\) & \(\boldsymbol{y_2=2 \times 10^{-5}}\) & \(\boldsymbol{y_2=10^{-5}}\) 
        &  
        & \(\boldsymbol{y_2=10^{-4}}\) & \(\boldsymbol{y_2=2 \times 10^{-5}}\) & \(\boldsymbol{y_2=10^{-5}}\) \\
        \hline
        0 & 4602 & 46861 & 39294 & 40440 & 4602 & 49322 & 41140 & 42247 & 4602 & 51882 & 43119 & 44174 \\
        1 & 9229 & 52320 & 59357 & 62567 & 9450 & 54417 & 61792 & 65156 & 9681 & 56634 & 64371 & 67900 \\
        2 & 13855 & 75760 & 82957 & 86230 & 14298 & 78771 & 86315 & 89744 & 14766 & 81955 & 89868 & 93466 \\
        3 & 18481 & 99569 & 106864 & 110175 & 19147 & 103509 & 111155 & 114625 & 19852 & 107674 & 115694 & 119334 \\
        4 & 23107 & 123607 & 130966 & 134304 & 23996 & 128486 & 136199 & 139697 & 24938 & 133641 & 141732 & 145402 \\
        5 & 27734 & 147800 & 155205 & 158562 & 28845 & 153623 & 161385 & 164903 & 30024 & 159777 & 167919 & 171609 \\
        6 & 32360 & 172105 & 179545 & 182916 & 33694 & 178878 & 186676 & 190209 & 35111 & 186034 & 194214 & 197921 \\
        7 & 36986 & 196495 & 203961 & 207343 & 38543 & 204221 & 212046 & 215592 & 40198 & 212383 & 220593 & 224312 \\
        8 & 41612 & 220951 & 228438 & 231829 & 43392 & 229633 & 237480 & 241035 & 45285 & 238804 & 247037 & 250767 \\
        9 & 46239 & 245460 & 252964 & 256362 & 48241 & 255100 & 262966 & 266528 & 50372 & 265283 & 273536 & 277273 \\
        \hline
    \end{tabular}}
    \caption{\textit{Comparison between the results from \cite{littlewood1997some} using a uniform (Beta) prior, and the results from using \acrshort*{cbi} priors. The \acrshort*{cbi} problem constraints, including the choice of $y_2$, are consistent with the Beta prior. The table lists the total number of demands needed (including $r$ unsuccessful demands) to justify a posterior reliability value of $(1-\alpha)$ for $m$ future successes. Here, $y_1=10^{-6}$, $p_1=y_1$,\, $p_2=y_2-y_1$ and $p_3=1-y_2$.} 
    }
    \label{tab_grouped_by_m_alpha_y2}
\end{sidewaystable}

Typically, to statistically demonstrate that safety-critical on-demand software is sufficiently reliable, such software is required to successfully handle a predetermined number of random demands without failure. If a failure occurs, one might think that the conservative thing to do is to stop the software execution, fix the failure-causing fault in the software, then restart the software demonstration ignoring the software's past performance on the previous demands. However,  \cite{littlewood1997some} show how it's more conservative to take the failure and past successes into account, and assume the fix has not changed the software's reliability, when deciding how many more demands the software should be subjected to. Using a uniform (Beta) prior distribution of \acrshort*{pfd}, $\text{B}(1,1)$, they compute the total number of demands needed to demonstrate a posterior probability, $(1-\alpha)$, of the software correctly handling $m$ future demands, upon observing the software fail $r$ and successfully handle $k$ past demands. Table~\ref{tab_grouped_by_m_alpha_y2} compares these demand numbers with what Theorem~\ref{thm_gen_sol} stipulates they should be, using partial specifications of the uniform prior in terms of $3$ intervals; for consistency with the uniform prior, $p_1=y_1$, $p_2=y_2-y_1$, $p_3=1-y_2$.

\begin{figure}[htbp!]
    \begin{minipage}[]{0.47\linewidth}
	\includegraphics[width=\linewidth]{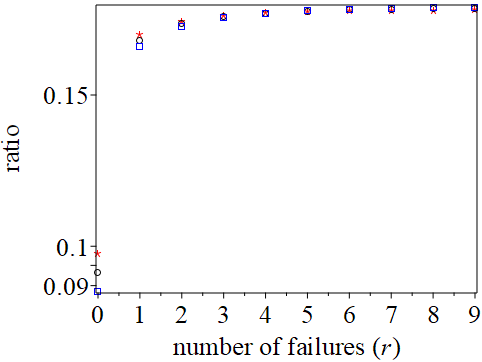}
	\caption{\textit{Ratio of the number of demands needed to demonstrate $(1-\alpha)$ reliability bound with $y_1=10^{-6}$, $y_2=10^{-4}$.} 
    }
	\label{fig_comp_1996_paper1}
\end{minipage}\hfill
\begin{minipage}[]{0.47\linewidth}
	\includegraphics[width=\linewidth]{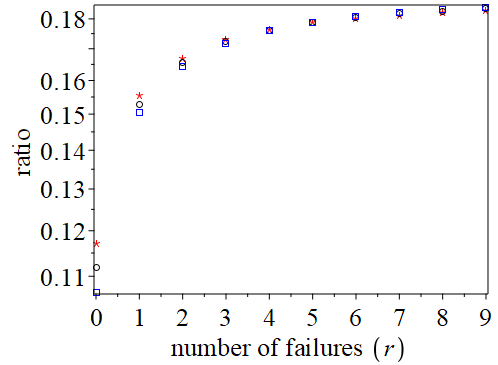}
	\caption{\textit{Same parameters as Figure~\ref{fig_comp_1996_paper1} except $\boldsymbol{y_2 = 2\times10^{-5}}$.\\}
    }
	\label{fig_comp_1996_paper2}
\end{minipage}\hfill
\hfill
\begin{minipage}{0.47\linewidth}  
	\includegraphics[width=\linewidth]{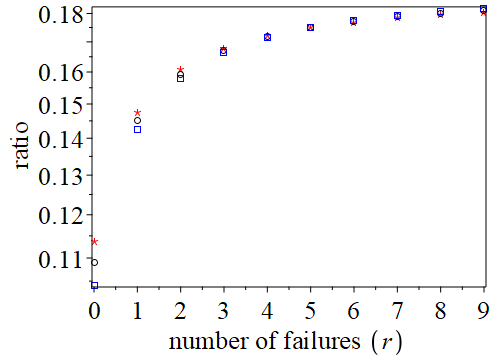}
	\caption{\textit{Same parameters as Figure~\ref{fig_comp_1996_paper1} except $\boldsymbol{y_2 = 10^{-5}}$
    }
    }
	\label{fig_comp_1996_paper3}
\end{minipage}\hfill
\begin{minipage}{0.43\linewidth}  
	\includegraphics[width=\linewidth]{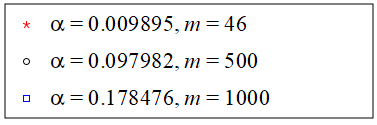}\vspace{2.0cm}
\end{minipage}
\end{figure}

The triplet of $(m,\alpha)$ values in the table~---~i.e. $(46,0.009895),\ (500, 0.097982)$ and $(1000,0.178476)$~---~are chosen for illustration because, under the Beta prior, they each require the software to successfully handle $4602$ demands without failure. This illustrates a range of claims that can be made with a fixed demand budget using the Beta prior. Contrastingly, the conservative \acrshort*{cbi} prior changes, depending on $m$, $\alpha$, and the constraining intervals in \eqref{eqn_morenoetalCBIprob}, so that different demand budgets are required even when no failures are observed. For instance, with $y_2=10^{-4}$, the required number of demands without failure under \acrshort*{cbi} is $46861$, $49322$, and $51882$. 


The table shows how the level of conservatism matters: i.e. the total number of required demands (with failures) is noticeably different, depending on whether the Beta prior, versus a \acrshort*{cbi} prior, is used. For the (small) numbers of failures in the table, the required number of demands under \acrshort*{cbi} is roughly an order of magnitude larger than the required number for the Beta prior. Using the table data, Figures~\ref{fig_comp_1996_paper1}, \ref{fig_comp_1996_paper2} and \ref{fig_comp_1996_paper3} depict how each ratio~---~the number of demands for the Beta prior divided by the number for a \acrshort*{cbi} prior~---~grows as the number of failures increases. In each figure, each ratio monotonically approaches an upper bound and stabilizes after only a few failures.

The upper bounds suggest conservatism still matters when $r$ is large. If $k_{\beta}$ and $k_C$ are the required amounts of failure-free operation (according to the Beta and \acrshort*{cbi} priors), to demonstrate the software continues to meet the $(1-\alpha)$ reliability requirement as $r$ increases, numerical estimation suggests $\lim\limits_{r\to\infty}k_{\beta}/r<\lim\limits_{r\to\infty}k_{C}/r<\infty$. Analytically, as $r\to\infty$, $k_{\beta}\to\infty$ and $\frac{r}{r+k_{\beta}}\to1-(1-\alpha)^{1/m}$ with the Beta prior. With a \acrshort*{cbi} prior, the asymptotic behavior of $r/(r+k_C)$ is more interesting. As $r$ and (therefore) $k_C$ increase, $\phi^*$ remains equal to the value $(1-\alpha)$, while $\frac{r}{r+m+k_{C}}$ and $\frac{r}{r+k_{C}}$ tend to the same value. The $h$ function tends to $(1-x)^m$, except at the limit of $\frac{r}{r+k_C}$ where $h$ will be undefined. The relationship in Figure~\ref{fig_h_stationarypoints}, among the fixed-point triplet, implies the difference ($h-\phi^*$) has limiting roots: either 
$y_*\to (1-{\phi^*}^{\frac{1}{m}})$ and $y_{**}\to\frac{r}{r+k_{C}}$ (e.g. see Figure~\ref{fig_convergence_of_h}), or the limits are switched. In general, $k_\beta/r$ and $k_C/r$ have different limits (see Appendices~D, E, Supplementary Material); see numerical examples of $y_*$ and $y_{**}$ convergence in Figures~\ref{fig_comp_stationary_points_1}, \ref{fig_comp_stationary_points_2}, \ref{fig_comp_stationary_points_3}, plotted using $(m,\alpha)$ from Table~\ref{tab_grouped_by_m_alpha_y2}.
%
%
%
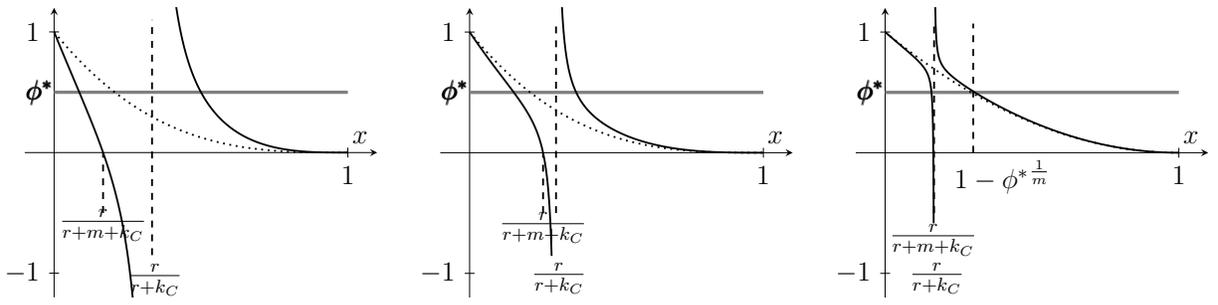
\begin{figure}[htbp!]
\centering
\begin{minipage}{0.32\textwidth}
\def\rval{1}
\def\mval{3}
\def\kval{2}
\def\xsing{\fpeval{\rval / (\rval + \kval)}}
\def\xroot{\fpeval{\rval / (\rval + \mval + \kval)}}

\resizebox{\linewidth}{!}{%
\begin{tikzpicture}
  \begin{axis}[height=6cm, domain=0:1, samples=400,
    xlabel={$x$}, axis lines=middle,
    xtick distance=1, ytick distance=1,
    enlargelimits=true, tick style={black}, ymin=-1, ymax=1
  ]
    \draw[line width=0.5mm,gray] ({axis cs:0,0.5}) -- ({axis cs:1,0.5});
    \node[black] at (axis cs:0.03,0.5) [anchor=east] {$\pmb{\phi^*}$};

    \addplot[domain=0:\fpeval{\xsing - 0.001}, thick, black]
    {(1 - x)^\mval * (\rval - x*(\mval + \kval + \rval)) / (\rval - x*(\kval + \rval))};
    \addplot[domain=\fpeval{\xsing + 0.001}:0.999, thick, black]
    {(1 - x)^\mval * (\rval - x*(\mval + \kval + \rval)) / (\rval - x*(\kval + \rval))};

    \draw[dashed, thick] ({axis cs:\xsing,-0.85}) -- ({axis cs:\xsing,1.1});
    \node[black] at (axis cs:\xsing+0.01,-1.05) {$\frac{r}{r+k_C}$};
    \draw[dashed, thick] ({axis cs:\xroot,-0.5}) -- ({axis cs:\xroot,0});
    \node[black] at (axis cs:\xroot,-0.6) {$\frac{r}{r+m+k_C}$};

    \addplot[domain=0:1, dotted, thick] {(1 - x)^\mval};
  \end{axis}
\end{tikzpicture}
}
\end{minipage}
\hfill
\begin{minipage}{0.32\textwidth}
\def\rval{5}
\def\mval{3}
\def\kval{12}
\def\xsing{\fpeval{\rval / (\rval + \kval)}}
\def\xroot{\fpeval{\rval / (\rval + \mval + \kval)}}

\resizebox{\linewidth}{!}{%
\begin{tikzpicture}
  \begin{axis}[height=6cm, domain=0:1, samples=400,
    xlabel={$x$}, axis lines=middle,
    xtick distance=1, ytick distance=1,
    enlargelimits=true, tick style={black}, ymin=-1, ymax=1
  ]
    \draw[line width=0.5mm,gray] ({axis cs:0,0.5}) -- ({axis cs:1,0.5});
    \node[black] at (axis cs:0.03,0.5) [anchor=east] {$\pmb{\phi^*}$};

    \addplot[domain=0:\fpeval{\xsing - 0.015}, thick, black]
    {(1 - x)^\mval * (\rval - x*(\mval + \kval + \rval)) / (\rval - x*(\kval + \rval))};
    \addplot[domain=\fpeval{\xsing + 0.001}:0.999, thick, black]
    {(1 - x)^\mval * (\rval - x*(\mval + \kval + \rval)) / (\rval - x*(\kval + \rval))};

    \draw[dashed, thick] ({axis cs:\xsing,-0.5}) -- ({axis cs:\xsing,1.1});
    \node[black] at (axis cs:\xsing+0.01,-1.04) {$\frac{r}{r+k_C}$};
    \draw[dashed, thick] ({axis cs:\xroot,-0.5}) -- ({axis cs:\xroot,0});
    \node[black] at (axis cs:\xroot,-0.62) {$\frac{r}{r+m+k_C}$};

    \addplot[domain=0:1, dotted, thick] {(1 - x)^\mval};
  \end{axis}
\end{tikzpicture}
}
\end{minipage}
\hfill
\begin{minipage}{0.32\textwidth}
\def\rval{20}
\def\mval{2}
\def\kval{100}
\def\xsing{\fpeval{\rval / (\rval + \kval)}}
\def\xroot{\fpeval{\rval / (\rval + \mval + \kval)}}

\resizebox{\linewidth}{!}{%
\begin{tikzpicture}
  \begin{axis}[height=6cm, domain=0:1, samples=400,
    xlabel={$x$}, axis lines=middle,
    xtick distance=1, ytick distance=1,
    enlargelimits=true, tick style={black}, ymin=-1, ymax=1
  ]
    \draw[line width=0.5mm,gray] ({axis cs:0,0.5}) -- ({axis cs:1,0.5});
    \node[black] at (axis cs:0.03,0.5) [anchor=east] {$\pmb{\phi^*}$};

    \addplot[domain=0:\fpeval{\xsing - 0.0015}, thick, black]
    {(1 - x)^\mval * (\rval - x*(\mval + \kval + \rval)) / (\rval - x*(\kval + \rval))};
    \addplot[domain=\fpeval{\xsing + 0.001}:0.999, thick, black]
    {(1 - x)^\mval * (\rval - x*(\mval + \kval + \rval)) / (\rval - x*(\kval + \rval))};

    \draw[dashed, thick] ({axis cs:\xsing,-0.5}) -- ({axis cs:\xsing,1.1});
    \node[black] at (axis cs:\xsing+0.01,-1.04) {$\frac{r}{r+k_C}$};
    \draw[dashed, thick] ({axis cs:\xroot,-0.1}) -- ({axis cs:\xroot,0});
    \node[black] at (axis cs:\xroot,-0.72) {$\frac{r}{r+m+k_C}$};

    \addplot[domain=0:1, dotted, thick] {(1 - x)^\mval};

    \draw[dashed, thick] ({axis cs:0.3,0}) -- ({axis cs:0.3,1.07}); node[below=5pt] {\space};
   \node[black] at (axis cs:0.4,-0.2) {${1-{\phi^*}}^{\frac{1}{m}}$};
  \end{axis}
\end{tikzpicture}
}
\end{minipage}

\caption{\textit{Example $h$ converging to $(1-x)^m$ almost everywhere as $r\to\infty$ and $k_C\to\infty$.}
}
\label{fig_convergence_of_h}
\end{figure}
\begin{figure}[htbp!]
    \begin{minipage}[]{0.47\linewidth}
	\includegraphics[width=\linewidth]{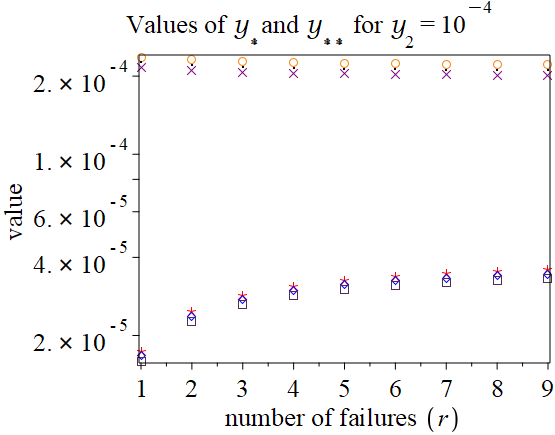}
	\caption{\textit{For each $(m,\alpha)$ in Table~\ref{tab_grouped_by_m_alpha_y2}, as $r$ and $k_C$ grow to infinity, $y_{**}\to\frac{r}{r+k_C}$ and $y_{*}\to 1-(1-\alpha)^{1/m}$.}}
	\label{fig_comp_stationary_points_1}
\end{minipage}\hfill
\begin{minipage}[]{0.47\linewidth}
	\includegraphics[width=\linewidth]
    {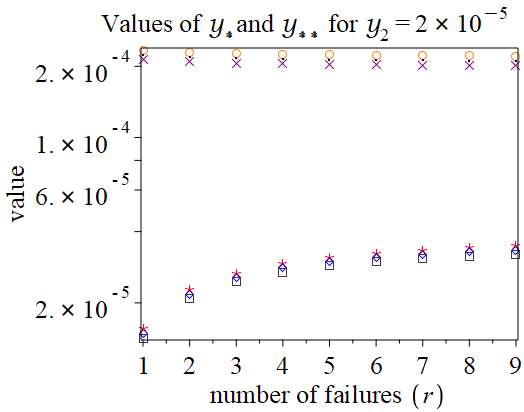}
	\caption{\textit{The same parameters as Figure~\ref{fig_comp_stationary_points_1} except $\boldsymbol{y_2 = 2\times10^{-5}}$.\\}}
	\label{fig_comp_stationary_points_2}
\end{minipage} 
\begin{minipage}{0.47\linewidth} 
	\includegraphics[width=\linewidth]{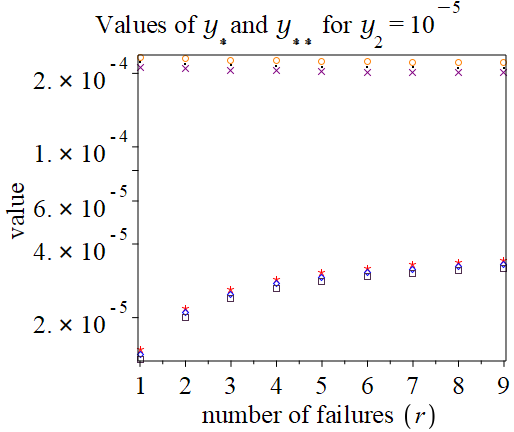}
	\caption{\textit{The same parameters as Figure~\ref{fig_comp_stationary_points_1} except $\boldsymbol{y_2 = 10^{-5}}$}}
	\label{fig_comp_stationary_points_3}
\end{minipage}\hfill
\begin{minipage}{0.37\linewidth}  
	\includegraphics[width=\linewidth]{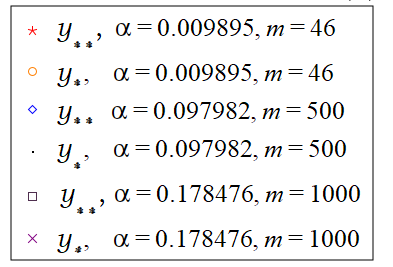}\vspace{2cm}
\end{minipage}
\end{figure}

As expected, in Table\ \ref{tab_grouped_by_m_alpha_y2}, the number of required demands increases with an increase in failures, under both the Beta and \acrshort*{cbi} priors. However, unlike the Beta prior, under \acrshort*{cbi}, there are $y_2$ values for which \emph{any} failure implies \emph{no amount} of failure-free operation is enough to demonstrate the $(1-\alpha)$ posterior probability~---~even an infinite amount won't suffice! This is true for any $y_2>1-(1-\alpha)^{\frac{1}{m}}$, because the aforementioned upper limit of $\phi^*$ is $(1-y_2)^m$ as $k\to\infty$; here, $\phi^*=1-\alpha$. Table\ \ref{tab_grouped_by_m_alpha_y2} excludes such $y_2$ values.

We have highlighted trends in Table~\ref{tab_grouped_by_m_alpha_y2} that do not change when $y_2$ is altered
. Although the number of demands appears to increase with decreasing $y_2$ when there \emph{are} failures, the number of demands is \emph{not} monotonic w.r.t. $y_2$ when there are \emph{no} failures. For instance, when $m=46$ and $\alpha=0.009895$, as $y_2$ decreases from $10^{-4}$ to $10^{-5}$, the total number of demands reduces from $46861$ to $39394$, then increases to $40440$. This non-monotonic behavior is unsurprising, since altering $y_2$ alters intervals $I_2$, $I_3$, and probabilities $p_2$, $p_3$.

There are situations where the software being free from faults is a real possibility (see  \cite{zhao_conservative_2015,zhao_conservative_2018},  \cite{strigini_software_2013}). In such cases, when no failures are observed, the appropriate \acrshort*{cbi} solution is the limit of Corollary~\ref{cor_special_case_no_failure_solution} as $y_1\to 0$, with probability $p_1$ of ``fault-freeness'': i.e. the limit
\begin{align}
\phi^*=\frac{p_1 + \sum_{i=2}^{j-1}(1-y_i)^{m+k}p_i + (1-y_*)^{m+k}p_j + \sum_{i=j+1}^{n}(1-y_{i-1})^{m+k}p_i}{p_1 + \sum_{i=2}^{j-1}(1-y_i)^kp_i + (1-y_*)^kp_j + \sum_{i=j+1}^{n}(1-y_{i-1})^kp_i}\,.
\label{eqn_CBI_faultfree}
\end{align}Appendix~F of the Supplementary Material derives this limit. \cite{strigini_software_2013} consider a special case of \eqref{eqn_CBI_faultfree}, with one constraining point and one constraining interval; i.e. the constraints $\mathbb P(X=0)=p_1$ and $\mathbb P(0<X\leqslant 1)=1-p_1$. 

Strigini and Povyakalo remark that any observed failure implies $\phi^*=0$. This is an instance of the observation we made earlier in this section: in general, at least three constraining intervals and the probability of fault-freeness are required, to prevent $\phi^*$ being necessarily zero. In fact, in the limiting ``fault-freeness'' case of Theorem~\ref{thm_gen_sol}, i.e. in the limit of the Theorem\ \ref{thm_gen_sol} solution as $y_1\to 0$, the form of $\phi^*$ is unchanged except for probability masses $p_1$ and $p_2$ being assigned to $x=0$. Thus, only constraining intervals $I_3$, and above, contribute non-trivially to $\phi^*$. As such, with only two intervals and the ``fault-freeness'' probability, the objective function simply becomes $\phi=(1-x_3)^m$, which tends to $0$ as $x_3$ tends to $1$. Hence, $\phi^*=0$. Intuitively, to be conservative, the assessor rules out the software being fault-free if any failure is observed, and considers any failure extremely strong evidence of the software being arbitrarily close to completely unreliable (i.e. arbitrarily close to the \acrshort*{pfd} being $1$).

Otherwise, with at least three intervals and a ``fault-freeness'' probability, $\phi^*$ grows as more failure-free operation is observed. Figures~\ref{fig_comp_2013_paper} and \ref{fig_comp_2013_paper1} compare how quickly $\phi^*$ grows with increasing $k$ (keeping $m$ fixed)\footnote{Since the ratio $m/k$ is used to generate the graphs, the graphs also illustrate how $\phi^*$ reduces as the required number of successful future demands $m$ increases (keeping $k$ fixed).}, comparing when there are no failures (from Strigini et al. with one constraining interval) with when there are \emph{some} (with three constraining intervals). The $p_i$ probabilities are chosen for illustration as $p_1=0.9$, $p_2=0.09$, $p_3=0.009$, and $p_4=0.0001$. Figure\ \ref{fig_comp_2013_paper} uses $y_1=0$, $y_2=10^{-6}$, $y_3=10^{-5}$, while Figure\ \ref{fig_comp_2013_paper1} uses the less stringent $y_1=0$, $y_2=10^{-5}$ and $y_3=10^{-3}$.

\begin{figure}[htbp!]
    \begin{minipage}[]{0.475\linewidth}
	\includegraphics[width=0.95\linewidth]{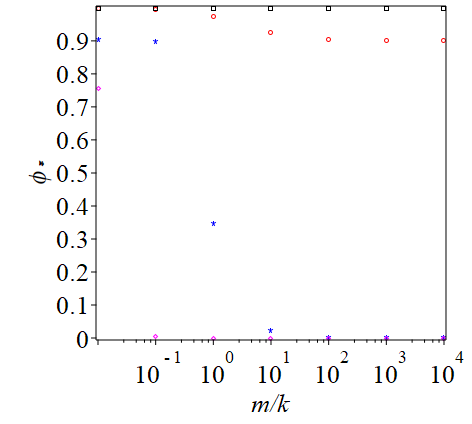}
	\caption{\textit{$\phi^*$ as successes $k$ increase from right to left with, and without, failures.}
    }
	\label{fig_comp_2013_paper}
\end{minipage}\hfill
\begin{minipage}{0.475\linewidth}
	\includegraphics[width=0.95\linewidth]{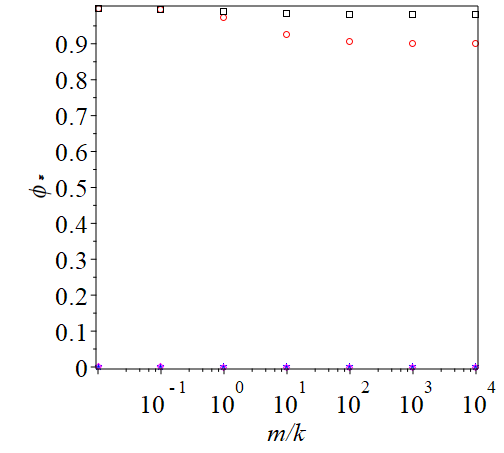}
	\caption{\textit{The same parameters as those in Figure~\ref{fig_comp_2013_paper} with less stringent $y_2$, $y_3$}}
	\label{fig_comp_2013_paper1}
\end{minipage}\hfill\centering\begin{minipage}{0.3\linewidth} 
	\includegraphics[width=0.9\linewidth]{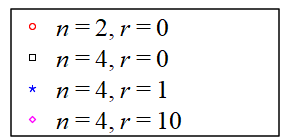}
\end{minipage}
\end{figure}

In these Figures the limits of $\phi^*$ differ for large $k$, depending on whether failures are observed or not. When no failures are observed, using $\phi^*$ from \eqref{eqn_CBI_faultfree}, $\phi^*$ grows to $1$ as $k\to\infty$ because of the probability of ``fault-freeness'' $p_1$. However, with failures, $\phi^*$ grows to $(1-y_3)^m$ instead, since $p_1$ and $p_2$ do not contribute to $\phi^*$ and the $y_i$ are ordered. In Figure~\ref{fig_comp_2013_paper}, $(1-y_3)^m$ is $0.905$, so the $r=1$ and $r=10$ curves grow to this limit. In Figure~\ref{fig_comp_2013_paper1}, these curves imperceptibly grow to the limit $4.5\times10^{-5}$. 

Similarly, as $m\to\infty$ with no failures, $\phi^*\to p_1/(p_1+p_2+\sum_{i=3}^{n}(1-y_{i-1})^kp_i)$. That is, for large $m$, $\phi^*$ shrinks to the smallest posterior probability of the software being fault-free. This follows from the limit of \eqref{eqn_CBI_faultfree} and the fact that $y_{*}\to0$ here (for a special case derivation, see  \cite{zhao_conservative_2015,zhao_conservative_2018}): the posterior probability is made as small as possible by making the denominator as large as possible, placing the $p_i$s at the lower endpoints of their respective intervals. In contrast, as $m\to\infty$ with failures, the $\phi^*$ in the ``fault-freeness'' limit of Theorem\ \ref{thm_gen_sol} tends to $0$.

\subsection{Limitations and Future Work}
\label{subsec_limitationsandFutureWork}
Like traditional Bayesian inference, \acrshort*{cbi} also requires some belief elicitation. The use of partial specifications of priors in \acrshort*{cbi} can be viewed as a ``conservative'' act~---~by design, this tries to prevent the use of unjustified beliefs for inference. Although this significantly reduces the traditional burden of eliciting a fully specified prior, some elicitation  questions remain under \acrshort*{cbi}, such as ``What is the best way to elicit partial specifications of priors, based on diverse forms of evidence?''. Can this form of (conservative) inference be shown to be more, or less, sensitive to certain forms of evidence? Such questions have been partially addressed in previous work (e.g.  \cite{zhao_conservative_2018,littlewood_reliability_2020},  \cite{salako_conservative_2021}), but more work is needed. One possible direction is the use of imprecise probabilities for the partial prior specifications, replacing the precise $p_i$ probabilities with intervals $[\underline{p_i},\overline{p_i}]$. See  \cite{walley1991statistical} and  \cite{Augustin2014} for general imprecise probability overviews, while  \cite{Coolen2004},  \cite{Utkin2007}, focus on its application to (continuously operating software) reliability assessment. 

Our example \acrshort*{cbi} application is the assessment of ``on-demand'' software, which executes only when needed (see  \cite{LittlewoodStrigini_1993}, \cite{Innal2010-PFD-PFH}); although we expect analogous results for continuously operating software, effects specific to continuous operation may yield meaningfully different outcomes.

Under fixed-point iteration, all the \acrshort*{cbi} problems\ \eqref{eqn_morenoetalBer} we have checked numerically quickly converge to the solutions. The fixed-point condition $h(y_*)=\phi^*$, along with attractor–repeller consistent prior-domain locations, are the fixed-point solution of the parametric problem $\sum_{i=1}^n \min_{x_i\in I_i}\{p_ix_i^r(1-x_i)^{k}((1-x_i)^m-\phi)\}$ where $\phi$ is the objective function in \eqref{eqn_morenoetalBer}. Dinkelbach-type iteration converges to $\phi^*$ from any start $\hat{\phi}_1\in(0,1]$: \textbf{1)} given $\hat{\phi}_t$, set $x_{j_2}=h^{-1}(\hat{\phi}_t)$ on $(\frac{r}{r+k},1]$;
\textbf{2)} for $i\neq j_2$, pick $x_i\in\{y_{i-1},y_i\}$ minimizing $p_ix_i^r(1-x_i)^{k}((1-x_i)^m-\hat{\phi}_t)$;
\textbf{3)} set $\hat{\phi}_{t+1}=\frac{\sum_{i=1}^{n}x_i^r(1-x_i)^{m+k} p_i}{\sum_{i=1}^{n}x_i^r(1-x_i)^{k} p_i}$\,. 
``Strict monotonicity'', ``root existence'', and continuity arguments, suggest that fixed-point iteration converges to conservative \acrshort*{cbi} solutions for fairly general \acrshort*{cbi} problems; proving this and characterizing the convergence speed will be useful (see \cite{Dinkelbach1967}, \cite{Schaible1976Dinkelbach}).


Problem\ \eqref{eqn_genCBIprob} provides conservative answers to ``single-objective'' questions; e.g. ``How large can the probability of future failures be, given the evidence?''. However, ``reliability'' is just one aspect of the wider notion of system ``dependability''~---~in practice, assessments of software ``security'', ``maintainability'', ``performability'', and ``timeliness'', might also require conservative answers  (see  \cite{AvizLapRand_2004} for taxonomy). Even when restricting solely to questions of ``reliability'', there are multiple questions that could be asked simultaneously and multiple ways of quantitatively answering these questions collectively via inference. Future work will look to broadening problem\ \eqref{eqn_genCBIprob} to cater for ``multi-objective'' conservative assessments.

Finally, demonstrating software meets ultra-high reliability requirements is often infeasible (see  \cite{ButlerFinelli_1993}). While Bayesian inference can alleviate this (see  \cite{LittlewoodStrigini_1993,LittlewoodStrigini_2011}), conservative inference under \acrshort*{cbi} will require partial prior specifications that strongly influence the inference. Identifying situations where such specifications can be justified is useful~---~see  \cite{zhao_assessing_2019,zhao_assessing_2020},  \cite{BishopStrigini_2022}, for examples.

\bibliography{bibliography}

\newpage
\appendix

\phantomsection\label{supplementary-material}
\bigskip

\begin{center}

{\large\bf SUPPLEMENTARY MATERIAL}

\end{center}

\begin{description}
\item[Title:]
Proofs of mathematical results in the paper ``Conservative Software Reliability Assessments Using Collections of Bayesian Inference Problems''.
\item[Appendix A:]
Proof of Proposition~\ref{prop_cbi_transform};
\item[Appendix B:]
Proof of Theorem~\ref{thm_gen_sol};
\item[Appendix C:]
Proof that the infimum $\phi^*$ of Theorem~\ref{thm_gen_sol} is monotonically increasing and bounded above; 
\item[Appendix D:]
Proof that the ratio of the number of demands needed to demonstrate a reliability bound~---~using a Beta prior vs using a \acrshort*{cbi}-based prior~---~is bounded; 
\item[Appendix E:]
Proof of the asymptotic behavior of the fixed-point that defines the solution of Theorem~\ref{thm_gen_sol};
\item[Appendix F:]
Proof of the functional form of the solution in Corollary~\ref{cor_special_case_no_failure_solution} in the limit of justifying a non-zero probability of the software being fault-free. 
\end{description}
\newpage
\section{Proof of Proposition\ \ref{prop_cbi_transform}}
\label{sec_app_prop1}
The \acrshort*{lhs} and \acrshort*{rhs} are well-defined, since they have non-zero denominators. So, suppose the claim of the proposition is not true. Then,
     $\dfrac{\mathbb E[f(X)]}{\mathbb E[g(X)]}\,\neq\,\frac{\sum\limits_{i=1}^{n}f(x_i)\mathbb P(A_i)}{\sum\limits_{i=1}^{n}g(x_i)\mathbb P(A_i)}$ for all $x_i\in I_i$. The \textit{r.h.s.} of the ``$\neq$'' is a non-negative, continuous function of $x_1\in I_1$, $\ldots$, $x_n\in I_n$. So, by the \acrfull*{ivt}, either 
    \begin{align*}
        \frac{\mathbb E[f(X)]}{\mathbb E[g(X)]}\,>\,\frac{\sum_{i=1}^{n}f(x_i)\mathbb P(A_i)}{\sum_{i=1}^{n}g(x_i)\mathbb P(A_i)}
    \qquad\text{ or}\qquad\frac{\mathbb E[f(X)]}{\mathbb E[g(X)]}\,<\,\frac{\sum_{i=1}^{n}f(x_i)\mathbb P(A_i)}{\sum_{i=1}^{n}g(x_i)\mathbb P(A_i)}
    \end{align*}exclusively, for all $x_1\in I_1$, $\ldots$, $x_n\in I_n$. W.l.o.g., suppose the ``$>$'' relationship holds. Then, $\mathbb E[f(X)]\sum\limits_{i=1}^{n}g(X_i)\mathbb P(A_i)\,>\,\mathbb E[g(X)]\sum\limits_{i=1}^{n}f(X_i)\mathbb P(A_i)$ 
    for random variables $X_i$ such that $X_i:A_i\rightarrow \mathbb R$ and $X_i(\omega):=X(\omega)$ for $\omega\in A_i$ (see Figure \ref{fig_cbi_transform}). 
\begin{figure}[htbp!]
\begin{center}
\begin{tikzpicture}[thick,scale=0.7, every node/.style={scale=0.7}]
    \begin{scope}[scale=0.5,rotate=125]
    
         \draw[clip, preaction={fill = lightgray, draw=white}] plot [smooth cycle, tension=1] coordinates {(0,5) (1,4.6) (2,4) (3,3.8) (5,0) (3.6,-4) (0,-5) (-1,-3.5) (-3,-4) (-4,0) (-3,2) (-1.5,4.3)};

        \begin{scope}[rotate=45,scale=1]
        \draw[fill = lightgray, thick, draw=white, line width = 0.07cm] plot [smooth cycle, tension=1] coordinates {(4,-3) (3,2) (4,4) (-1,4.5) (1,0.5) (-3,-1) (-8,0.5)};
        \draw[fill = lightgray, thick, draw=white, line width = 0.07cm] plot [smooth cycle, tension=1] coordinates {(-3,-8) (0,-1) (2,0) (1,3) (0,8) (-6,10)};
        \draw[fill = lightgray, thick, draw=white, line width = 0.07cm] plot [smooth cycle, tension=1] coordinates {(0,8) (-1,1) (-4,2) (-6, 1) (-6,10)};
        \draw[fill = lightgray, thick, draw=white, line width = 0.07cm] plot [smooth cycle, tension=1] coordinates {(-3,7) (3,1) (3.5,3) (2,4) (6,8)};
        \end{scope}
        
        \begin{scope}[scale=0.4,shift={(4.2,7)},rotate around={90:(4.2,7)}] 
        \draw[fill = lightgray, thick, draw=white, line width = 0.07cm] plot [smooth cycle, tension=1] coordinates {(0,8) (-1,1) (-4,2) (-6, 1) (-6,10)};
        \draw[fill = lightgray, thick, draw=white, line width = 0.07cm] plot [smooth cycle, tension=1] coordinates {(-3,7) (3,1) (3.5,3) (2,4) (6,8)};    
        \end{scope}
        
        \begin{scope}[scale=0.5,shift={(3,0)},rotate around={90:(3,0)}]        
        \draw[fill = lightgray, thick, draw=white, line width = 0.07cm] plot [smooth cycle, tension=1] coordinates {(0,8) (-1,1) (-4,2) (-6, 1) (-6,10)};
        \draw[fill = lightgray, thick, draw=white, line width = 0.07cm] plot [smooth cycle, tension=1] coordinates {(-3,-8) (0,-1) (2,0) (1,3) (0,8) (-6,10)};
        \draw[fill = lightgray, thick, draw=white, line width = 0.07cm] plot [smooth cycle, tension=1] coordinates {(-3,7) (3,1) (3.5,3) (2,4) (6,8)};    
        \end{scope}

        \begin{scope}[scale=0.3,shift={(-8,4)},rotate around={45:(-8,4)}]
        \draw[fill = lightgray, thick, draw=white, line width = 0.07cm] plot [smooth cycle, tension=1] coordinates {(0,8) (-1,1) (-4,2) (-6, 1) (-6,10)};
        \draw[fill = lightgray, thick, draw=white, line width = 0.07cm] plot [smooth cycle, tension=1] coordinates {(-3,7) (3,1) (3.5,3) (2,4) (6,8)};    
        \end{scope}

        \begin{scope}[scale=0.5,shift={(-5,-4)},rotate around={25:(-5,-4)}]
        \draw[fill = lightgray, thick, draw=white, line width = 0.07cm] plot [smooth cycle, tension=1] coordinates {(0,8) (-1,1) (-4,2) (-6, 1) (-6,10)};
        \draw[fill = lightgray, thick, draw=white, line width = 0.07cm] plot [smooth cycle, tension=1] coordinates {(-3,7) (3,1) (3.5,3) (2,4) (6,8)};    
        \end{scope}
    \end{scope}

        \begin{scope}
            \node[scale=2.5] (A) at (-1.5,3) {$\pmb{\Omega}$};
            \node[align=center] (B1) at (11,3.1) {$\pmb{X_2(\omega)=X(\omega)}$\\ for $\pmb{\omega\in A_2}$};
            \node[align=center] (B2) at (8.3,1.2) {$\pmb{X_1(\omega)=X(\omega)}$\\ for $\pmb{\omega\in A_1}$};
            \node[align=center] (B3) at (5.85,-0.9) {$\pmb{X_i(\omega)=X(\omega)}$\\ for $\pmb{\omega\in A_i}$};

        
            \node[scale=1.3] (C) at (0.55,1.8) {$\pmb{A_2}$}; 
            \draw[line width=0.04cm,-{Latex[length=5mm]}] plot [smooth, tension=0.6] coordinates { (0.9, 1.5) (3,1.7) (9,3) (12.04,0.35) };

            \node[scale=1.3] (D) at (0.47,0.5) {$\pmb{A_1}$}; 
            \draw[line width=0.04cm,-{Latex[length=5mm]}] plot [smooth, tension=0.75] coordinates { (0.9, 0.5) (3,0.9) (9,2) (10.44,0.05) };            

            \node[scale=1.3] (E) at (1.7,-0.55) {$\pmb{A_i}$}; 
            \draw[line width=0.04cm,-{Latex[length=5mm]}] plot [smooth, tension=0.55] coordinates { (2.1,-0.55) (6.5,-1.9) (8,-1.7) (7,0.2) (11.5,-2.4) (13.7,-0.6) };                        
        \end{scope}

        \begin{scope}
        \draw[line width=0.06cm, -{Latex[length=5mm]}] (9,-0.25) -- (15.5,-0.25);

        \node[scale=1.1] (A) at (10.5,-0.8) {$\pmb{X(A_1)}$};
        \node[scale=1.1] (B) at (12.05,-0.8) {$\pmb{X(A_2)}$};
        \node[scale=2] (D) at (15.5,-0.8) {$\mathbb R$};
        \node[scale=1.1] (E) at (13.71,0.2) {$\pmb{X(A_i)}$};
        \node[scale=1.1] (E) at (13.73,-0.5) {$\pmb{\underbrace{}{}}$};
        

        \node[scale=2.4] (IU) at (12.1,0.15) {$\pmb{\overbrace{}{}}$};
        \node[scale=1] (IL) at (10.47,0) {$\pmb{\overbrace{}{}}$};
        
        \end{scope}        
\end{tikzpicture}
\end{center}\caption{\textit{Illustrating the $X_i$ mappings. Note, by definition of the $A_i$, $X(A_i)=I_i$.}} \label{fig_cbi_transform}
\end{figure}
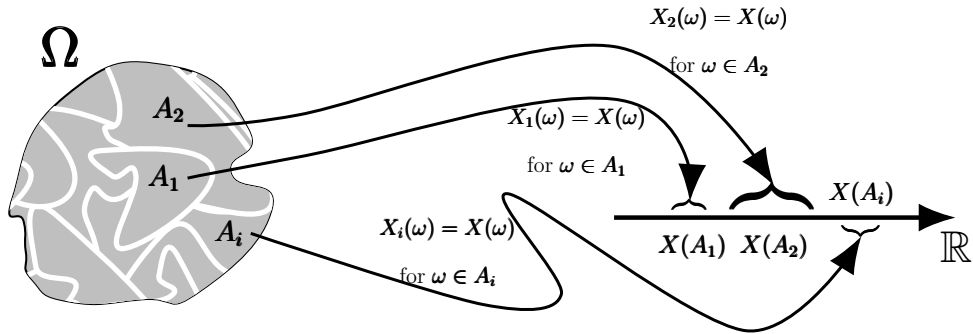

 Consider probability measures\footnote{These $\mu_i$s are well-defined: $\Sigma\cap A_i$ is a sigma-algebra of subsets of $A_i$; a consequence of $A_i$ belonging to $\Sigma$, and $\Sigma$ being a sigma-algebra. Also note, $\mathbbm{1}_{\mathtt S}$ is an indicator function that equals $1$ if logical predicate $\mathtt S$ is true, and equals $0$ otherwise.} $\mu_i:\Sigma\cap A_i\rightarrow [0,1]$, $\mu_i(A):=\int\limits_{\Omega}\mathbbm{1}_{\omega\in A}\mathtt{d}\mathbb P(\omega)/\mathbb P(A_i)$ for any $A\in \Sigma\cap A_i$, and the probability space $(\times_{i=1}^{n}A_i, \otimes_{i=1}^{n}\Sigma \cap A_i,\otimes_{i=1}^{n}\mu_i)$. The ``$>$'' inequality involving the $X_i$s is between $\otimes_{i=1}^{n}\Sigma \cap A_i/B(\mathbb R)$-measurable functions. Hence, taking expectations (w.r.t. the product probability measure) of both sides of the inequality, 
and by the linearity of expectations, 
\begin{align}
\label{eq_cbi_transform_ineq}
     &\mathbb E[f(X)]  \sum\limits_{i=1}^{n}\ \int\limits_{\times_{j=1}^{n}A_j}\!\!\!\!g(X_i)\,\mathtt{d}\!\otimes_{j=1}^{n}\mu_j\,\mathbb P(A_i) \hspace{0.1cm}>\mathbb E[g(X)]\sum\limits_{i=1}^{n}\int\limits_{\times_{j=1}^{n}A_j}\!\!\!\! f(X_i)\,\mathtt{d}\!\otimes_{j=1}^{n}\mu_j\,\mathbb P(A_i)\,.
     \end{align} 

We make the following three observations:
   \begin{enumerate}
       \item 
     $\int\limits_{\times_{j=1}^{n}A_j}\!\!\!\! g(X_i)\mathtt{d}\!\otimes_{j=1}^{n}\mu_j\,=\,\int\limits_{A_1} \mathtt{d}\mu_1 \cdots \int\limits_{A_i}g(X_i) \mathtt{d}\mu_i\cdots\int\limits_{A_n} \mathtt{d}\mu_n\,,\ \text{by Tonelli's Theorem}$;
\item 
 $\int\limits_{A_i}g(X_i) \mathtt{d}\mu_i =\int\limits_{A_i}g(X)|_{A_i} \mathtt{d}\mu_i=\frac{\int\limits_{\Omega}g(X)\mathbbm{1}_{A_i} \mathtt{d}\mathbb P}{\mathbb P(A_i)}=\frac{\mathbb E[g(X)\mathbbm{1}_{A_i}]}{\mathbb P(A_i)}$, by $X_i$ and $\mu_i$ definitions; 
 \item $\int\limits_{A_i} \mathtt{d}\mu_i=1, \,\text{by definition of } \mu_i$\,.
   \end{enumerate}
   
Similar observations apply to $f(X_i)$. Using these, inequality \eqref{eq_cbi_transform_ineq} becomes $$\mathbb E[f(X)] \sum_{i=1}^{n}\frac{\mathbb E[g(X)\mathbbm{1}_{A_i}]}{\mathbb P(A_i)}\mathbb P(A_i)\,>\,\mathbb E[g(X)] \sum_{i=1}^{n}\frac{\mathbb E[f(X)\mathbbm{1}_{A_i}]}{\mathbb P(A_i)}\mathbb P(A_i)\,.$$Hence, we arrive at the contradiction, $\mathbb E[f(X)]\mathbb E[g(X)]>\mathbb E[f(X)]\mathbb E[g(X)]$.

\newpage
\section{Proof of Theorem~\ref{thm_gen_sol}}
\label{app_proofofCBIproblem}
A three-step proof of Theorem~\ref{thm_gen_sol}: 
\begin{enumerate}
   \item[]\textbf{step I:} the gradient of the objective function determines the functional forms the objective function can assume when minimized. Prove that two of the gradient's components each have a non-trivial root in $[0,1]$ when these components equal zero, then use these roots and the signs of the gradient's components to deduce finitely many functional forms for the objective function when minimized.
  \item[]\textbf{step II:} prove these functional forms have finitely many infima, thus the objective function has finitely many infima: the smallest infimum is the global infimum.
  \item[]\textbf{step III:} based on the previous steps, state the infimum, state \eqref{eqn_morenoetalBer}'s solution as a fixed-point system, and state the prior that gives the infimum.
\end{enumerate}

\noindent\textbf{Step I:}
There are a finite number of plausible functional forms that $\phi$ can take when minimized. To deduce these, first, we prove the existence of a pair of points at which two components of the gradient of $\phi$ are zero. Then, using this pair to deduce the signs of gradient components, we deduce the functional forms.

Define $f,g:[0,1] \longrightarrow [0,1]$, $f(x)=(1-x)^{m}g(x)$ and $g(x)=x^r(1-x)^{k}$, with $k>0,\,r > 0$. Consider $x_i\in I_i$ for $i=1\ldots n$. The  gradient of $\phi$ has $i$-th component
\begin{equation}
     \frac{\partial \phi}{\partial x_{i}}=\frac{g(x_i)(r-x_i(r+k))p_i}{x_i(1-x_i){\sum_{j=1}^{n}g(x_j)p_{j}}}\left(h(x_i)-\phi\right)\,,
     \label{eq_grad}
 \end{equation} where $h$ is the function $h:[0,1]\setminus\{\frac{r}{r+k}\}\longrightarrow \mathbb R$, $h(x)=(1-x)^m\left(\frac{r-x(m+k+r)}{r-x(k+r)}\right)$. From \eqref{eq_grad}, the $i$-th gradient component is non-trivially zero if, and only if, 
\begin{equation}    
h(x_i)=\phi
\label{eq_stationarity_cond}
\end{equation}
Otherwise, the gradient component is positive or negative, depending on combinations of whether $h(x_i)>\phi$ or $h(x_i)<\phi$, and $x_i<\frac{r}{r+k}$ or $x_i>\frac{r}{r+k}$.

The signs of the gradient components determine a preferred location in each compact interval $\bar{I}_i$ ($\bar{I}_i$ is the closure of $I_i$) where probability mass should be assigned to minimize $\phi$. This restricts the possible functional forms of $\phi$ that equal the infimum, $\phi^*$, of $\phi$. 
Probability $p_i$ could be placed at $y_{i-1}$ if $\partial\phi/\partial x_i>0$ in a neighborhood of $y_{i-1}$, or it could be placed at $y_i$ if $\partial\phi/\partial x_i<0$ in a neighborhood of $y_{i}$. It could also be placed at $x$, if $x$ is a stationary point (i.e. $\partial\phi/\partial x_i=0$ at $x_i=x$) at which $\phi$ has a local infimum. 

Whenever the $p_i$s are assigned within each $\bar{I}_i$, there are only two stationary points in $[0,1]$ that satisfy $\partial \phi/\partial x_i=0$~---~one point which is a local infimum, and the other a local supremum. This follows from the graph of $h$, which we deduce now.

Since $h$ is a rational function, it is continuously differentiable over its domain. The derivative of $h$ shows $h$ is monotonically decreasing over $[0,\frac{r}{r+m+k})$ and $(\frac{r}{r+k},1]$, and any stationary points of $h$ must lie in $(\frac{r}{r+m+k},\frac{r}{r+k})$. Differentiating $h$ for $x\neq \frac{r}{r+k}$, \begin{equation}
    \dfrac{\mathtt{d} h}{\mathtt{d}x}=(1-x)^m\left(\dfrac{r-x(r+m+k)}{r-x(r+k)}\right)'+\left(\dfrac{r-x(r+m+k)}{r-x(r+k)}\right)\left((1-x)^m\right)'
    \label{eq_grad_h}
\end{equation} 
    Observe, $\left((1-x)^m\right)'=-m(1-x)^{m-1}<0$ and $\left(\frac{r-x(r+m+k)}{r-x(r+k)}\right)'=\frac{-mr}{(r(1-x)-kx)^2}<0$. Thus, $\mathtt{d} h/\mathtt{d}x<0$ for $x<\frac{r}{r+m+k}$, because $r-x(r+m+k)>0$ and ${r-x(r+k)}>0$ whenever $x<\frac{r}{r+m+k}$. Similarly, $\mathtt{d} h/\mathtt{d}x<0$ for $x>\frac{r}{r+k}$, because $r-x(r+m+k)<0$ and ${r-x(r+k)}<0$ whenever $x>\frac{r}{r+k}$. So, $h$ is monotonically decreasing over these intervals, as claimed. However, for $\frac{r}{r+m+k}< x< \frac{r}{r+k}$ (i.e. $r-x(r+m+k)<0$ and ${r-x(r+k)}>0$), the sign of $\mathtt{d} h/\mathtt{d}x$ depends on $r,\ k$, and $m$.
    
    Interestingly, $h(x)<0$ over the interval  $\frac{r}{r+m+k}< x< \frac{r}{r+k}$. Furthermore, $h$ has, at most, a pair of stationary points in this interval; because, from \eqref{eq_grad_h}, $\mathtt{d} h/\mathtt{d}x=0$ non-trivially if, and only if,\begin{equation*}x=\dfrac{2r^2+(2k+m+1)r\pm\sqrt{-4rk^2-4kr(m+r)+r^2(m-1)^2}}{2(r+k)(r+m+k)}\end{equation*}
These stationary points are real \emph{iff} $(-4rk^2-4kr(m+r)+r^2(m-1)^2)\geqslant0$. Otherwise, there are no stationary points: the points are a complex conjugate pair. 

When real, this pair of points lie in $(\frac{r}{r+m+k},\frac{r}{r+k})$. The smaller of the two points is greater than $\frac{r}{r+m+k}$, because 
\begin{align*}
    &\dfrac{2r^2+(2k+m+1)r-\sqrt{-4rk^2-4kr(m+r)+r^2(m-1)^2}}{2(r+k)(r+m+k)}>\frac{r}{r+m+k}\\
    &\iff r(m+1)-\sqrt{-4rk^2-4kr(m+r)+r^2(m-1)^2}>0
\end{align*} The last inequality is true because $r(m+1)>r(m-1)$. 
Analogously, the larger of the two points is less than $\frac{r}{r+k}$, because 
\begin{align*}
    &\dfrac{2r^2+(2k+m+1)r+\sqrt{-4rk^2-4kr(m+r)+r^2(m-1)^2}}{2(r+k)(r+m+k)}<\frac{r}{r+k}\\
    &\iff \dfrac{-r(m-1)+\sqrt{-4rk^2-4kr(m+r)+r^2(m-1)^2}}{2(r+k)(r+m+k)}<0\\
    &\iff -r(m-1)+\sqrt{-4rk^2-4kr(m+r)+r^2(m-1)^2}<0
\end{align*} The last inequality is true because the radicand subtracts the positive number $4rk^2+4kr(m+r)$ from $r^2(m-1)^2$, so the square-root of the radicand must be smaller than $r(m-1)$. We conclude, both stationary points of $h$ lie in $(\frac{r}{r+m+k},\frac{r}{r+k})$.

 So far, we have shown $h$ is monotonically decreasing  on $[0,\frac{r}{r+m+k})$ and $(\frac{r}{r+k},1]$. Notice, from the definition of $h$, $h(x)\geqslant 0$ over these intervals. We have also shown $h$ has a pair of stationary points (if any) in $(\frac{r}{r+m+k},\frac{r}{r+k})$. And, over this interval, $h<0$ by definition. Altogether, this means \eqref{eq_stationarity_cond} has a root~---~$\phi$ can equal $h$~---~only where $h$ is positive and monotonically decreasing; that is, only over $[0,1]\setminus(\frac{r}{r+m+k},\frac{r}{r+k})$. In fact, there are two reasons why $0<\phi<1$, so that $\phi$ is bounded by $h$:
 \textbf{i)} the terms in $\phi$'s numerator are strictly less than the corresponding terms in its denominator; \textbf{ii)} both numerator and denominator are strictly positive. 
 
 Figure~\ref{fig_function_h} illustrates the graph of $h$. There is an $x^*\in(\frac{r}{r+k},1]$ such that $h(x^*)=1$. So, $h$ is $0$ at $1$ and $\frac{r}{r+m+k}$, while $h$ is $1$ at $0$ and $x^*$. 
 Moreover, for an arbitrary value of $\phi$, denoted $\hat{\phi}$, this value is represented as a horizontal line at height $\hat{\phi}$. And, for any $x_u\in[x^*,1]$ such that $h(x_u)=\hat{\phi}$, a corresponding $x_l\in[0,\frac{r}{r+m+k}]$ also satisfies $h(x_l)=\hat{\phi}$. Since $h$ bounds $\phi$, and the horizontal line of height $\hat{\phi}$ intersects $h$ at exactly two points, these imply the functional form of $\phi$ when minimized must be based on exactly two components of the gradient of $\phi$ vanishing. The local supremum acts like a ``source'' that ``repels'' probability masses in nearby intervals, while the local infimum acts like a ``sink'' that ``attracts'' probability masses. This proves that at most two components of the gradient have non-trivial roots.


\def\rval{1}
\def\mval{3}
\def\kval{2}

\def\xsing{\fpeval{\rval / (\rval + \kval)}}

\def\xroot{\fpeval{\rval / (\rval + \mval + \kval)}}

\begin{figure}[htbp!]
\begin{center}
\scalebox{0.85}
{\begin{tikzpicture}
  \begin{axis}[
    width=13cm,
    height=8cm,
    domain=0:1,
    samples=400,
    xlabel={$x$},
    ylabel={$h(x)$},
    axis lines=middle,
    xtick distance=1,
    ytick distance=1,
    enlargelimits=true,
    tick style={black},
    ymin=-1, ymax=1,
  ]

  \draw[line width=1mm,gray] ({axis cs:0,0.7}) -- ({axis cs:1,0.7});
  \node[black] at (axis cs:-0.01,0.7) [anchor=east] {$\pmb{\hat{\phi}}$};

  \addplot[
    domain=0:\fpeval{\xsing - 0.001},
    line width = 0.5mm,
    black,
  ]
  {(1 - x)^\mval * (\rval - x*(\mval + \kval + \rval)) / ( \rval - x*(\kval + \rval) )};

  \addplot[
    domain=\fpeval{\xsing + 0.001}:0.999,
    line width = 0.5mm,
    black,
  ]
  {(1 - x)^\mval * (\rval - x*(\mval + \kval + \rval)) / ( \rval - x*(\kval + \rval) )};


  \draw[dashed,thick,black] ({axis cs:\xsing,-0.9}) -- ({axis cs:\xsing,1.1});
  \node[black] at (axis cs:\xsing+0.003,-1.06) {$\pmb{\frac{r}{r+k}}$};

  \draw[dashed,thick,black] ({axis cs:\xsing+0.095,-0.3}) -- ({axis cs:\xsing+0.1,1});
  \node[black] at (axis cs:\xsing+0.105,-0.4)  {$\pmb{x^\ast}$};

  \draw[dashed,thick,black] ({axis cs:\xroot,-0.5}) -- ({axis cs:\xroot,0});
  \node[black] at (axis cs:\xroot,-0.65) {$\pmb{\frac{r}{r+m+k}}$};

  \draw[dashed,thick,black] ({axis cs:0.05,0.7}) -- ({axis cs:0.05,0});
  \node[black] at (axis cs:0.06,-0.1) {$\pmb{x_l}$};

  \draw[dashed,thick,black] ({axis cs:0.465,0.7}) -- ({axis cs:0.465,0});
  \node[black] at (axis cs:0.479,-0.1) {$\pmb{x_u}$};

  \end{axis}
\end{tikzpicture}}
\end{center}\caption{\textit{Illustration of the graph of $h(x) = (1 - x)^m\left(\frac{r - x(m + k + r)}{r - x(k + r)}\right)$.}} \label{fig_function_h}
\end{figure}
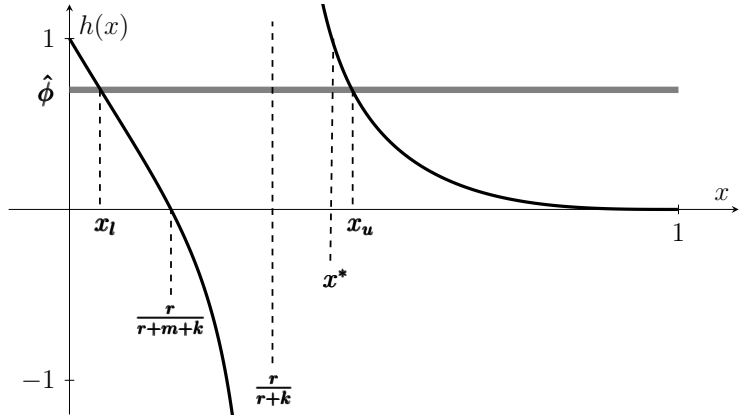

 Consequently, to minimise $\phi$, all of the $x_i$ must be assigned in their respective $I_i$ as dictated by a ``source'' and ``sink'' pair. All $x_i$ to the left of the ``source'' must be located at the lower endpoints of their respective $I_i$ intervals, while $x_i$ between the ``source'' and  ``sink'' must be located at the upper endpoints of their $I_i$. The $x_i$ for the interval containing the ``source'' must be located at one of the endpoints of the interval~---~whichever endpoint gives the lower objective function value. Analogous reasoning in terms of the ``sink'' requires all $x_i$ to the right of the ``sink'' to be located at the lower endpoints of their respective $I_i$. The $x_i$ in the same interval as the ``sink'' should be located \emph{at} the ``sink''.
 


In interval $I_{1}$, $x_1$ must be $0$ when $n\geqslant 2$. Indeed, express $\phi$ as a function of $\alpha$, \begin{equation*}
    \phi(\alpha)=(1-x_1)^{m}\dfrac{g(x_1)p_1}{\sum_{i=1}^ng(x_i)p_i}+\sum_{j=2}^{n}(1-x_j)^{m}\dfrac{g(x_j)p_j}{\sum_{i=1}^n g(x_i)p_i}=(1-x_1)^{m}\alpha+(1-\alpha)\beta
\end{equation*} where  $\alpha=g(x_1)p_1/\sum_{i=1}^n g(x_i)p_i$ and $\beta=\sum_{j=2}^{n}(1-x_j)^{m}g(x_j)p_j/\sum_{i=2}^n g(x_i)p_i$. Since $(1-x_1)^m\geqslant(1-x_2)^m\geqslant\dots\geqslant(1-x_n)^m$ then, for all $\alpha\in[0,1]$, we have $\phi(\alpha)\geqslant \beta$. Thus, $\phi(\alpha)\geqslant \phi(0)$, and $\phi$ is smallest necessarily when $x_1=0$ in $\bar{I}_1$. Otherwise, when $n=1$ instead, $\phi$ is $(1-x_1)^m$, so $x_1=1$ minimizes $\phi$. 

All of the foregoing justifies there being only two possible functional forms, $\tilde{\phi}_1$ and $\tilde{\phi}_2$, for $\phi$ at its infima. Define $x^*_1$ to be $\max\{x^*,y_1\}$ and choose $x\in[x^*_1,1]$. Denote the interval containing $x$ as the ${j_2}^{th}$ interval, and the ${j_1}^{th}$ interval as the interval containing the corresponding $x_l$ under $h$. Then, $\tilde{\phi_1}:[x^*_1,1]\longrightarrow[0,1]$ and $\tilde{\phi_2}:[x^*_1,1]\longrightarrow[0,1]$ are the functional forms,  \begin{align*}  &\tilde{\phi_1}(x):=
    \begin{cases}
  \dfrac{\sum\limits_{i=2}^{j_{1}}f(y_{i-1})p_i+\sum\limits_{i=j_{1}+1}^{j_{2}-1}f(y_{i})p_i+f(x)p_{j_2}+\sum\limits_{i=j_{2}+1}^{n}f(y_{i-1})p_i}{\sum\limits_{i=2}^{j_{1}}g(y_{i-1})p_i+\sum\limits_{i=j_{1}+1}^{j_{2}-1}g(y_{i})p_i+g(x)p_{j_2}+\sum\limits_{i=j_{2}+1}^{n}g(y_{i-1})p_i}\,,\hfill j_1<j_2\\
     \dfrac{\sum\limits_{i=2}^{n}f(y_{i-1})p_i}{\sum\limits_{i=2}^{n}g(y_{i-1})p_i}\,,\hfill j_1=j_2
\end{cases}
\\
  &\tilde{\phi_2}(x):=\begin{cases}
   \dfrac{\sum\limits_{i=2}^{j_1-1}\!\!f(y_{i-1})p_i+\sum\limits_{i=j_{1}}^{j_{2}-1}\!f(y_{i})p_i+f(x)p_{j_2}+\sum\limits_{i=j_2+1}^{n}\!\!f(y_{i-1})p_i}{\sum\limits_{i=2}^{j_1-1}\!\!g(y_{i-1})p_i+\sum\limits_{i=j_{1}}^{j_{2}-1}\!g(y_{i})p_i+g(x)p_{j_2}+\sum\limits_{i=j_2+1}^{n}\!\!g(y_{i-1})p_i}\,,\hfill  j_1<j_2\\
    \dfrac{\sum\limits_{i=2}^{j_{2}-1}f(y_{i-1})p_i+f(x)p_{j_2}+\sum\limits_{i=j_{2}+1}^{n}f(y_{i-1})p_i}{\sum\limits_{i=2}^{j_{2}-1}g(y_{i-1})p_i+g(x)p_{j_2}+\sum\limits_{i=j_{2}+1}^{n}g(y_{i-1})p_i}\,,\hfill  j_1=j_2 \end{cases}
\end{align*} 
Notice how $j_1$ and $j_2$ can change as $x$ changes, since the domain $[x_1^*,1]$ of $x$ values can intersect multiple $I_i$ intervals. Also note that, since $x_1=0$, we have $f(x_1)=f(0)=0$ and $g(x_1)=g(0)=0$, so these terms do not appear in the functional forms.

\noindent\textbf{Step II:}
$\phi$ has a finite number of local infima. This is equivalent to only a finite number of possible placements of the $p_i$ probabilities in the $\bar{I}_i$ intervals, in such a way that the infimum $\phi^*$ is either $\tilde{\phi}_1$ or $\tilde{\phi}_2$ while, simultaneously, $\phi^*$ satisfies $h(x_i)=\phi^*$ for some $x_i\in\bar{I}_i$. To show this, consider all possible placements of the $p_i$ consistent with the following three cases:

\noindent Let the interval containing $x^*$ be $I_{i^*}$.
\begin{enumerate}
    \item[] \textbf{case I:} Assume $x$ and $x_1^*$ lie in separate intervals. For each $i\leqslant i^*$ place $p_i$ at an endpoint of $I_i$. Otherwise, place $p_i$ as close as possible to $x$ for the other intervals above $I_{i^*}$; in particular, place $p_i$ at $x$ for the interval containing $x$ (see Figure~\ref {fig_mass_placement_case1}).

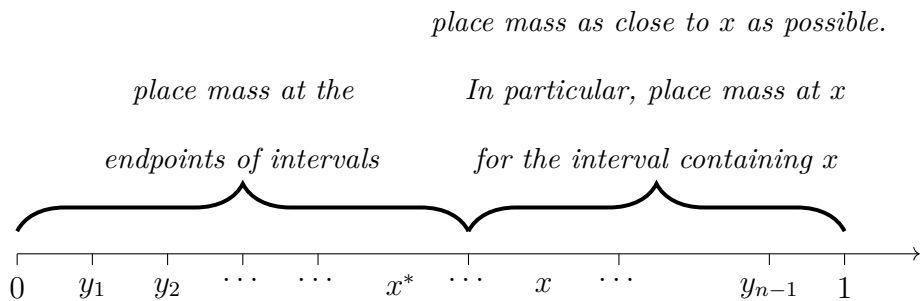
\begin{figure}[htbp!]\centering
\begin{tikzpicture}
  \draw[->] (0,0) -- (12,0);

 \draw (0,0) -- (0,-0.15) node[below=5pt] at (0,0) {$0$};
  \foreach \x/\label in {
    1/{$y_1$},
    2/{$y_2$},
    3/{\dots},
    4/{\dots},
    6/{\dots},
    8/{\dots},
    10/{$y_{n-1}$},
     11/{$1$}
  }{
    \draw (\x,0) -- (\x,-0.15);
    \node[below=5pt] at (\x,0) {\label};
  }

  \node[below=3.5pt] at (5.1,0) {$x^*$};

  \node[below=5pt] at (7,0) {$x$};

  \draw [decorate,decoration={brace,amplitude=18pt},yshift=0.3cm,line width=1.5pt]
    (0,0) -- (6,0) 
    node[midway,above=17pt,text width=4cm,align=center]
    {\small\textit{place mass at the endpoints of intervals}};
    
  \draw [decorate,decoration={brace,amplitude=18pt},yshift=0.3cm,line width=1.5pt]
    (6,0) -- (11,0) 
    node[midway,above=17pt,text width=6cm,align=center]
    {\small\textit{place mass as close to $x$ as possible. In particular, place mass at $x$ for the interval containing $x$}};
\end{tikzpicture}
    \caption{\textit{Illustration of  possible probability mass placement: case I}}\label{fig_mass_placement_case1}
\end{figure}
      \item[] \textbf{case II:} Assume $x$ and $x_1^*$ are in the same interval. For each $i< i^*$ place $p_i$ at an endpoint of $I_i$. Place $p_i$ as close as possible to $x$ in all other intervals; in particular, place $p_{i^*}$ at $x$ in interval $I_{i^*}$ containing $x$ and $x^*$ (see Figure~\ref {fig_mass_placement_case2}).
      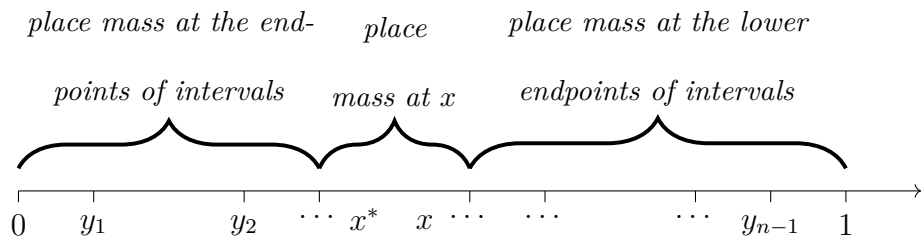
\begin{figure}[ht!]\centering
      \begin{tikzpicture}
  \draw[->] (0,0) -- (12,0);

   \draw (0,0) -- (0,-0.15) node[below=5pt] at (0,0) {$0$};
  \foreach \x/\label in {
    1/{$y_1$},
    3/{$y_2$},
    4/{\dots},
    6/{\dots},
    7/{\dots},
    9/{\dots},
    10/{$y_{n-1}$},
    11/{$1$}
  }{
    \draw (\x,0) -- (\x,-0.15);
    \node[below=5pt] at (\x,0) {\label};
  }

  \node[below=2.5pt] at (4.6,0) {$x^*$};
  
    \node[below=2.5pt] at (8,0) {\space};

  \node[below=5pt] at (5.4,0) {$x$};

  \draw [decorate,decoration={brace,amplitude=18pt},yshift=0.3cm,line width=1.5pt]
    (0,0) -- (4,0) 
    node[midway,above=19pt,text width=5.5cm,align=center]
    {\small\textit{place mass at the endpoints of intervals}};
  \draw [decorate,decoration={brace,amplitude=18pt},yshift=0.3cm,line width=1.5pt]
    (4,0) -- (6,0) 
    node[midway,above=19pt,text width=2cm,align=center]
    {\small\textit{place mass at $x$}};
   
  \draw [decorate,decoration={brace,amplitude=19pt},yshift=0.3cm,line width=1.5pt]
    (6,0) -- (11,0) 
    node[midway,above=19pt,text width=5.5cm,align=center]
    {\small\textit{place mass at the lower endpoints of intervals}};
\end{tikzpicture}
          \caption{\textit{Illustration of possible probability mass placement: case II}}
          \label{fig_mass_placement_case2}
      \end{figure}
      \item[] \textbf{case III:} Assume $x$ and $x_1^*$ are in the same interval. For each $i< i^*$ place mass at an endpoint of $I_i$. Place $p_{i^*}$ at the lower endpoint of the interval $I_{i^*}$ containing $x$ and $x^*$. Otherwise, place mass as close as possible to $x$ in the intervals above $I_{i^*}$ (see Figure~\ref {fig_mass_placement_case3}).
\begin{figure}[ht!]\centering 
   \begin{tikzpicture}
  \draw[->] (0,0) -- (12,0);

   \draw (0,0) -- (0,-0.15) node[below=5pt] at (0,0) {$0$};
  \foreach \x/\label in {
    1/{$y_1$},
    3/{$y_2$},
    4/{\dots},
    6/{\dots},
    7/{\dots},
    8/{\dots},
    10/{$y_{n-1}$},
    11/{$1$}
  }{
    \draw (\x,0) -- (\x,-0.15);
    \node[below=5pt] at (\x,0) {\label};
  }

  \node[below=2.5pt] at (4.6,0) {$x^*$};

  \node[below=5pt] at (5.4,0) {$x$};

  \draw [decorate,decoration={brace,amplitude=18pt},yshift=0.3cm,line width=1.5pt]
    (0,0) -- (4,0) 
    node[midway,above=17pt,text width=4.6cm,align=center]
    {\small\textit{place mass at the endpoints of intervals}};
   
  \draw [decorate,decoration={brace,amplitude=18pt},yshift=0.3cm,line width=1.5pt]
    (4,0) -- (11,0) 
    node[midway,above=17pt,text width=7cm,align=center]
    {\small\textit{place mass at the lower endpoints of intervals}};
\end{tikzpicture}

	\caption{\textit{Illustration of possible probability mass placement: case III} 
    }
	\label{fig_mass_placement_case3}
\end{figure}
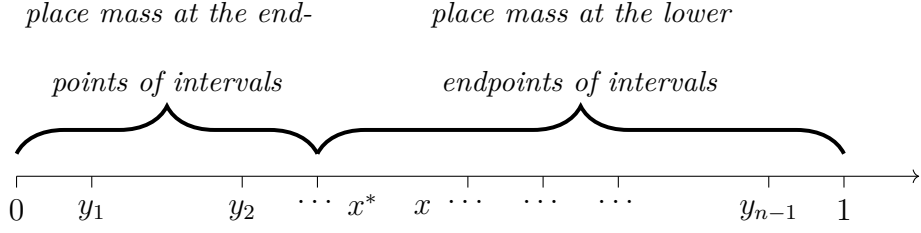
\end{enumerate}
    Placements of the $p_i$s that are consistent with these three cases define functional forms. For example, from case (I), define $\phi_{z_i,\ldots,z_{i^*},j_{2}}: \bar{I}_{j_2}\cap[x^*_1,1]\longrightarrow [0,1]$ as 
    \begin{align*}
    \phi_{z_i,\ldots,z_{i^*},j_2}(x)=\dfrac{\sum\limits_{i=1}^{i^*}f(z_i)p_i+\sum\limits_{i=i^*+1}^{j_{2}-1}f(y_i)p_i+f(x_i)p_{j_{2}}+\sum\limits_{i=j_{2}+1}^{n}f(y_{i-1})p_i}{\sum\limits_{i=1}^{i^*}g(z_i)p_i+\sum\limits_{i=i^*+1}^{j_{2}-1}g(y_i)p_i+g(x_i)p_{j_{2}}+\sum\limits_{i=j_{2}+1}^{n}g(y_{i-1})p_i}
    \end{align*}
where $j_2$ is an arbitrary choice from $(i^*+1),\ldots,n$, while $z_1,\ldots,z_{i^*}$ are an arbitrary placement of probability masses at an endpoint in each of the intervals $I_1,\ldots, I_{i^*}$. Analogous univariate rational functions can be defined based on cases (II) and (III). Altogether, there are a finite number of such rational functions. Note, $\tilde{\phi}_1$ and  $\tilde{\phi}_2$ must agree with some of these functional forms. So, the infimum of $\phi$ must also agree with some of these functional forms. 

Consider all such functional forms that agree on the locations in the intervals to the left of $x_1^*$; for example, the functions $\phi_{z_i,\ldots,z_{i^*},j_2}$ for all $j_2$. These define a piecewise continuous function $F:[x^*_1,1]\rightarrow [0,1]$, $F(x)=\phi_{z_i,\ldots,z_{i^*},j_2}(x)-h(x)$ where $x\in \bar{I}_{j_2}$. Since $F(x_1^*)$ is non-positive (because $h(x_1^*)\geqslant(1-x_1^*)^m\geqslant\phi_{z_i,\ldots,z_{i^*},i^*}(x_1^*)$ for $x\geqslant x_1^*$) and $F(1)$ is non-negative (because $h(1)=0$ and the $\phi_{z_i,\ldots,z_{i^*},j_2}$ are positive), the \acrshort*{ivt} implies that $F(x)=0$ for some $j_2$ interval containing this $x$. Moreover, at such an $x$, $\partial\phi/\partial x_{j_2}(x)$ is zero and $F$ must be increasing, since $F^\prime(x)=\phi^\prime_{z_i,\ldots,z_{i^*},j_2}(x)-h^\prime(x)=\partial\phi/\partial x_{j_2}(x)-h^\prime(x)=-h^\prime(x)>0$. If $x$ lies at an interval boundary, the same inequality holds for the appropriate one-sided derivative of $F$, and adjacent $F$ pieces glue continuously with consistent one-sided slopes; hence, the $F$ root is still simple. Thus, $F=0$ can only be satisfied once over $[x_1^*,1]$. Since there are finitely many such $F$, there are finitely many of these functional forms that satisfy $\phi_{z_i,\ldots,z_{i^*},j_2}=h$, so there are finitely many points in the domain of $\phi$ at which $\phi$ could be at a minimum. These points, up to a choice of endpoint in the $j_1$ interval, must give distinct $\phi$ values: since otherwise, if they have the same value $\hat{\phi}$ of $\phi$, they must have the same value of $\hat{x}\in[x_1^*,1]$ that solves $h(\hat{x})=\hat{\phi}$, and this $\hat{x}$ determines the locations in all other intervals up to a choice of endpoint in interval $j_1$~---~the same functional form follows from the same $\phi$ value. 

Finally, since $h$ is monotonically decreasing over $[x_1^*,1]$, the horizontal lines at the minima values of $\phi$ intersect $h$ at various heights, with only one unique value that has the lowest intersection with $h$: this is the global minimum $\phi^*$.

\noindent\textbf{Step III:} In summary, there exists the unique fixed point consisting of $\phi^*,\,y_{*},\,y_{**}$ that satisfies the system of equations $\phi^* = \underset{x\in[x^*,1]}{\inf} \{\tilde{\phi}(x)\}=\tilde{\phi}(y_{*})$ and $h(y_{**})=\phi^*=h(y_{*})$, for some $i$-indices $j_1,j_2$ such that $y_{**}\in\bar{I}_{j_1}$, $y_{*}\in\bar{I}_{j_2}$ and $y_{**}< \frac{r}{r+m+k}<\frac{r}{r+k}<y_{*}$. And, $\phi$ attains its infimum $\phi^*$ if the prior distribution of $X$ in Theorem\ \ref{thm_gen_sol} is used for inference.

\newpage
\section{Proof that $\phi^*\to(1-y_2)^m$ monotonically, as $k \to\infty$}
\label{sec_app_phistarconvergence}
Fix $m>0$, $r>0$, and consider $\tilde{\phi}_{1},\,\tilde{\phi}_{2}$ as functions of $m,\,k,\,r$ and $x$. For any $x\in[x^*,1]$, $\tilde{\phi}_{1}(m,k,r,x)$ and $\tilde{\phi}_{2}(m,k,r,x)$ are increasing functions of $k$. To see this, choose any small $\varepsilon>0$. One can show that $\tilde{\phi}_{1}(m,k+\varepsilon,r,x)>\tilde{\phi}_{1}(m,k,r,x)$ or, equivalently, that 
\begin{align}
&\left|\begin{array}{cc}\text{Num}(m,k+\varepsilon,r,x)&\text{Num}(m,k,r,x)\\\text{Den}(m,k+\varepsilon,r,x)&\text{Den}(m,k,r,x)\end{array}\right|>0&
\label{eqn_determinant_increasewithk}
\end{align}
where $\tilde{\phi}_{1}(m,k,r,x)=\frac{{\text{Num}}(m,k,r,x)}{{\text{Den}}(m,k,r,x)}$. The multi-linearity and antisymmetry of determinants, as well as the ordering and boundedness of the $ y_i$s, mean the determinant in \eqref{eqn_determinant_increasewithk} is a sum of strictly positive determinants. So, $\tilde{\phi}_{1}(m,k+\varepsilon,r,x)>\tilde{\phi}_{1}(m,k.,r,x)$ and, analogously, $\tilde{\phi}_{2}(m,k+\varepsilon,r,x)>\tilde{\phi}_{2}(m,k,r,x)$, hence \begin{equation*} \min\{\tilde{\phi}_{1}(m,k+\varepsilon,r,x),\tilde{\phi}_{2}(m,k+\varepsilon,r,x)\}>\min\{\tilde{\phi}_{1}(m,k,r,x),\tilde{\phi}_{2}(m,k,r,x)\}\end{equation*} by the monotonicity of ``$\min$''. Thus, $\tilde{\phi}$ is also monotonically increasing in $k$, and so is ${\phi}^{*}$; i.e. by the monotonicity of ``$\inf$'', $\underset{x\in[x^*,1]}{\inf}\tilde{\phi}(m,k+\varepsilon,r,x)\geqslant\underset{x\in[x^*,1]}{\inf}\tilde{\phi}(m,k,r,x)$.

$\phi^*$ is bounded, since $0<\phi<1$. Together, the \emph{completeness of real numbers}, as well as the boundedness and monotonicity of $\phi^\ast$, imply $\lim\limits_{k\rightarrow\infty}\phi^\ast$ exists. 

This limit is $(1-y_2)^m$, as the following argument shows. For a given value of $k$, we have $\phi^{*}_k:=\underset{x}{\inf}\tilde{ \phi}_{k}(x)$ and $\tilde{ \phi}_{k}(x):=\min\{\tilde{ \phi}_{1k}(x),\tilde{ \phi}_{2k}(x)\}$. We want to show that \begin{align}
\lim\limits_{k\to\infty}\phi^*_{k}&=\lim\limits_{k\to\infty}\underset{x}{\inf}\tilde{ \phi}_{k}(x)=\underset{x}{\inf}\lim\limits_{k\to\infty}\tilde{ \phi}_{k}(x)=\underset{x}{\inf}\{\min\{\lim\limits_{k\to\infty}\tilde{ \phi}_{1k}(x),\lim\limits_{k\to\infty}\tilde{ \phi}_{2k}(x)\}\}\nonumber\\&=\min\{\underset{x}{\inf}\lim\limits_{k\to\infty}\tilde{ \phi}_{1k}(x),\underset{x}{\inf}\lim\limits_{k\to\infty}\tilde{ \phi}_{2k}(x)\}=(1-y_2)^m  \label{eq_monotonicity_of_phi_star}
 \end{align}

The first equality in \eqref{eq_monotonicity_of_phi_star}
 holds by definition. The second equality, $\lim\limits_{k\to\infty}\underset{x}{\inf}\tilde{ \phi}_{k}(x)=\underset{x}{\inf}\lim\limits_{k\to\infty}\tilde{ \phi}_{k}(x)$, holds because:
 \begin{enumerate}
\item  $\lim\limits_{k\to\infty} \underset{x}{\inf} \tilde{ \phi}_{k}(x)\geqslant  \underset{x}{\inf}\lim\limits_{k\to\infty}\tilde{ \phi}_{k}(x) $: For each $k$, $\tilde{\phi}_k$ has a finite number of discontinuities. Choose small $\varepsilon>0$ and $\eta>0$. Enclose these discontinuities in sufficiently small open intervals for all $k$ such that \textbf{1)} the maximum total length of these intervals is less than $(1-x^*_1)\eta$ and, \textbf{2)} for each $k$, there exists $x_k$ outside of all these open intervals such that $\tilde{\phi}_k(x_k)<\inf\limits_{x}\tilde{\phi}_k(x)+\varepsilon$. By this construction, the subset of $[x^*_1,1]$ containing the sequence of $x_k$ is compact, and the sequence of $\tilde{\phi}_k$s are continuous over this subset. Thus, there is a convergent subsequence $x_{k_i}$ and a limit $\bar{x}$ of this sequence in this compact set so that, for fixed $j$, $\tilde{\phi}_{j}(x_{k_i})\to\tilde{\phi}_j(\bar{x})$. For fixed $j$, the monotonicity of $\tilde{\phi}_k$ implies $\tilde{\phi}_j(x_{k_i})\leqslant\tilde{\phi}_{k_i}(x_{k_i})<\inf\limits_{x}\tilde{\phi}_{k_i}(x)+\varepsilon$ for all $k_i\geqslant j$. Upon taking limits, previous observations about the convergent $x_{k_i}$ further imply $\tilde{\phi}_j(\bar{x})\leqslant\lim\limits_{k_i\to\infty}\tilde{\phi}_{k_i}(x_{k_i})\leqslant\lim\limits_{k_i\to\infty}\inf\limits_{x}\tilde{\phi}_{k_i}(x)+\varepsilon$, for fixed $j$. Since this holds for all $j$, it holds in the limit of large $j$; i.e. $\lim\limits_{j\to\infty}\tilde{\phi}_{j}(\bar{x})\leqslant\lim\limits_{k_i\to\infty}\inf\limits_{x}\tilde{\phi}_{k_i}(x)+\varepsilon$. Consequently, $\inf\limits_{x}\lim\limits_{k\to\infty}\tilde{\phi}_{k}(x)\leqslant\lim\limits_{k\to\infty}\tilde{\phi}_{k}(\bar{x})\leqslant\lim\limits_{k\to\infty}\inf\limits_{x}\tilde{\phi}_{k}(x)+\varepsilon$.
\item $\underset{x}{\inf}  \lim\limits_{k\to\infty}\tilde{ \phi}_{k}(x) \geqslant
   \lim\limits_{k\to\infty} \underset{x}{\inf} \tilde{ \phi}_{k}(x)$, since 
    $\lim\limits_{k\to\infty}\tilde{ \phi}_{k}(x) \geqslant
    \tilde{ \phi}_{k}(x)$ for all $k$ (thus, for all $k$, $\underset{x}{\inf}  \lim\limits_{k\to\infty}\tilde{ \phi}_{k}(x) \geqslant
    \underset{x}{\inf} \tilde{ \phi}_{k}(x)$).
\end{enumerate}

   It suffices to show $\lim\limits_{k\to\infty}\tilde{ \phi}_{k}(x)=\min\{\lim\limits_{k\to\infty}\tilde{ \phi}_{1k}(x),\lim\limits_{k\to\infty}\tilde{ \phi}_{2k}(x)\}$ for the third equality in \eqref{eq_monotonicity_of_phi_star}, which follows immediately from two facts: \textbf{1)} the $\tilde{\phi}_{ik}$s are monotonically increasing and bounded, so their limits exist; \textbf{2)} the function $m:\mathbb R^2\rightarrow\mathbb R$, $m(u,v):=\min(u,v)$ is continuous. 
   
   For the fourth equality, 
   any pair of real-valued functions over an interval, $f$ and $g$, satisfies $\underset{x}{\inf}\{\min\{f(x),g(x)\}\}=\min\{\underset{x}{\inf}f(x),\underset{x}{\inf}g(x)\}$. Therefore, \begin{equation*}\underset{x}{\inf}\{\min\{\lim\limits_{k\to\infty}\tilde{ \phi}_{1k}(x),\lim\limits_{k\to\infty}\tilde{ \phi}_{2k}(x)\}\}=\min\{\underset{x}{\inf}\lim\limits_{k\to\infty}\tilde{ \phi}_{1k}(x),\underset{x}{\inf}\lim\limits_{k\to\infty}\tilde{ \phi}_{2k}(x)\}\,.\end{equation*}
   
   For the fifth equality, recall that for any $x\in[x^*_1,1]$ there is $x_l\in[0,\frac{r}{r+m+k}]$ such that $h(x_l)=h(x)$. Also recall, $x_l<x^*<x$ by definition. Then one, and only one, of the following mutually exclusive situations holds:
   \begin{enumerate}
     \item If $x^*\in I_1$ then $x_1^*=y_1$ and $x_l$ must also be in $I_1$. Since the preferred location in the first interval is always $0$, this implies $\tilde{\phi}_{1}(x)=\tilde{\phi}_{2}(x)$. If $x\in I_2$ then $\lim\limits_{k\to\infty}\tilde{\phi}_{1}(x)=\lim\limits_{k\to\infty}\tilde{\phi}_{2}(x)=(1-x)^m\geqslant(1-y_2)^m$ because, for $x<y_2$, \begin{align*}
\lim\limits_{k\to\infty}\tilde{\phi}_{1}
=&\,\,
\frac{(1-x)^{m}p_2+\sum_{i=3}^{n}\left(\frac{y_{i-1}}{x}\right)^r\lim\limits_{k\to\infty}\left(\frac{1-y_{i-1}}{1-x}\right)^k(1-y_{i-1})^{m}p_{i}}{p_2+\sum_{i=3}^{n}\left(\frac{y_{i-1}}{x}\right)^r\lim\limits_{k\to\infty}\left(\frac{1-y_{i-1}}{1-x}\right)^kp_{i}}\\=&\,\,(1-x)^m\,, \text{ since the } y_i\text{s are strictly ordered.}\end{align*} If instead, $x\in I_j$ where $j>2$, then $\lim\limits_{k\to\infty}\tilde{\phi}_{1}(x)=\lim\limits_{k\to\infty}\tilde{\phi}_{2}(x)=(1-y_2)^m$. So, $\min\{\underset{x}{\inf} \,\lim\limits_{k\to\infty}\tilde{\phi}_{1}(x),\underset{x}{\inf}\,\lim\limits_{k\to\infty}\tilde{\phi}_{2}(x)\}=(1-y_2)^m$
      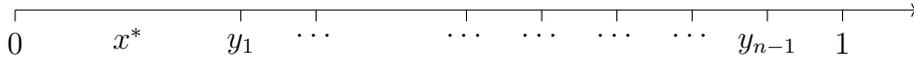
\begin{figure}[htbp!]\centering
	\begin{tikzpicture}
  \draw[->] (0,0) -- (12,0);

  \draw (0,0) -- (0,-0.15) node[below=5pt] at (0,0) {$0$};
  \foreach \x/\label in {
    3/{$y_1$},
    4/{\dots},
    6/{\dots},
    7/{\dots},
    8/{\dots},
    9/{\dots},
    10/{$y_{n-1}$},
    11/{$1$}
  }{
    \draw (\x,0) -- (\x,-0.15);
    \node[below=5pt] at (\x,0) {\label};
  }

  \node[below=2.5pt] at (1.5,0) {$x^*$};

\end{tikzpicture}
	\caption{ \textit{Illustration of $x^*$ lying in $I_1$}}
	\label{fig_comp_1996_pap1}
\end{figure}     
       \item  If $x^*\in I_2$ then $x_1^*=x^*$ and $x_l\in I_j$ (where $j\leqslant2$), depending on the value of $x$. If $x_l\in I_1$ and $x\in I_2$, then $\tilde{\phi}_1(x)=\tilde{\phi}_2(x)$ and $\lim\limits_{k\to\infty}\tilde{\phi}_{1}(x)=\lim\limits_{k\to\infty}\tilde{\phi}_{2}(x)=(1-x)^m\geqslant(1-y_2)^m$. If instead, $x_l\in I_1$ and $x\in I_j$ where $j>2$, then $\lim\limits_{k\to\infty}\tilde{\phi}_{1}(x)=\lim\limits_{k\to\infty}\tilde{\phi}_{2}(x)=(1-y_2)^m$. While, if $x_l\in I_2$ and $x\in I_j$ where $j>2$, then $\lim\limits_{k\to\infty}\tilde{\phi}_{1}(x)=(1-y_1)^m$ and $\lim\limits_{k\to\infty}\tilde{\phi}_{2}(x)=(1-y_2)^m$. Lastly, if $x_l,\,x\in I_2$, then $\lim\limits_{k\to\infty}\tilde{\phi}_{1}(x)=(1-y_1)^m$ and $\lim\limits_{k\to\infty}\tilde{\phi}_{2}(x)=(1-x)^m$. So, altogether, $\inf\limits_x\{\lim\limits_{k\to\infty}\tilde{\phi}_{2}(x)\}=(1-y_2)^m$. Therefore, $\min\{\underset{x}{\inf} \,\lim\limits_{k\to\infty}\tilde{\phi}_{1}(x),\underset{x}{\inf}\,\lim\limits_{k\to\infty}\tilde{\phi}_{2}(x)\}=(1-y_2)^m$.
       \begin{figure}[htbp!]
	\centering
	\begin{tikzpicture}
  \draw[->] (0,0) -- (12,0);

  \draw (0,0) -- (0,-0.15) node[below=5pt] at (0,0) {$0$};
  \foreach \x/\label in {
    1/{$y_1$},
    3/{$y_2$},
    4/{\dots},
    6/{\dots},
    7/{\dots},
    9/{\dots},
    10/{$y_{n-1}$},
    11/{$1$}
  }{
    \draw (\x,0) -- (\x,-0.15);
    \node[below=5pt] at (\x,0) {\label};
  }

  \node[below=2.5pt] at (2,0) {$x^*$};

\end{tikzpicture}
	\caption{
    \textit{Illustration of $x^*$ lying in $I_2$}}
	\label{fig_comp_1996_pap2}
\end{figure}
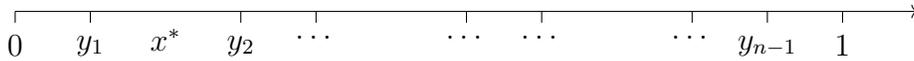      
       \item If $x^*\in I_i$ where $i>2$, then $x$ can be placed arbitrarily close to $x^*$ to ensure $x_l$ is arbitrarily close to $0$. Thus, $\lim\limits_{k\to\infty}\tilde{\phi}_{1}(x)=\lim\limits_{k\to\infty}\tilde{\phi}_{2}(x)=(1-y_2)^m$. Note, $\inf\limits_x\{\lim\limits_{k\to\infty}\tilde{\phi}_{1}(x)\}\geqslant(1-y_2)^m$ and $\inf\limits_x\{\lim\limits_{k\to\infty}\tilde{\phi}_{2}(x)\}\geqslant(1-y_2)^m$. So, $$\min\{\underset{x}{\inf} \,\lim\limits_{k\to\infty}\tilde{\phi}_{1}(x),\underset{x}{\inf}\,\lim\limits_{k\to\infty}\tilde{\phi}_{2}(x)\}=(1-y_2)^m\,.$$
       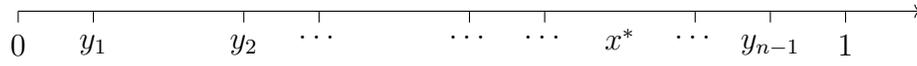
\begin{figure}[htbp!]
	\centering
	
\begin{tikzpicture}
  \draw[->] (0,0) -- (12,0);

  \draw (0,0) -- (0,-0.15) node[below=5pt] at (0,0) {$0$};
  \foreach \x/\label in {
    1/{$y_1$},
    3/{$y_2$},
    4/{\dots},
    6/{\dots},
    7/{\dots},
    9/{\dots},
    10/{$y_{n-1}$},
    11/{$1$}
  }{
    \draw (\x,0) -- (\x,-0.15);
    \node[below=5pt] at (\x,0) {\label};
  }

  \node[below=2.5pt] at (8,0) {$x^*$};

\end{tikzpicture}
	\caption{ \textit{Illustration of $x^*$ lying in $I_i$ for $i>2$.}}
	\label{fig_comp_1996_pap3}
\end{figure}
   \end{enumerate}

\newpage
\section{The ratio of the number of demands needed to demonstrate a reliability bound~---~using a Beta prior vs using a \acrshort*{cbi} prior~---~is asymptotically bounded.}
Let $k_{C}$ and $k_{\beta}$ denote the number of failure-free demands needed to demonstrate a posterior probability $(1-\alpha)$ of the software successfully handling $m$ future demands given a number of failures $r$, resulting from using the conservative \acrshort*{cbi} prior and Beta prior, respectively. In what follows, to ensure comparison between the Beta and \acrshort*{cbi} requirements, we need $1-(1-y_2)^m<\alpha<1$. Then, for a given $r$, the ratio of the total number of demands needed from using the two priors is $\frac {k_{\beta}+r}{k_{C}+r}$. This ratio is bounded as the number of failures increases (i.e. as $r\to\infty$). The following arguments prove this claim. 

 Show  $\lim\limits_{r\to\infty}\frac{k_{\beta}}{r}=\frac{(1-\alpha)^{1/m}}{{1-(1-\alpha)^{1/m}}}$ as follows. \cite{littlewood1997some} used a Beta prior $\text{Beta}(1,1)$, which is equivalent to a uniform distribution. 
Recall,  $$\mathbb P(m \text{ \emph{failure-free demands}} \mid k \text{\emph{ failure-free and }} r \text{\emph{ failed demands}})$$ is the posterior probability of interest. Using the Beta prior, this posterior probability is $(1-\alpha)$ if, and only if, $\frac{\text{Beta}(1+r,1+m+k_\beta)}{\text{Beta}(1+r,1+k_\beta)}=(1-\alpha)$. Simplifying, this is the requirement $\prod\limits_{i=1}^{m}\frac{k_\beta+i}{r+k_\beta+1+i}=(1-\alpha)\,.$ This implies $k_\beta/r$ is bounded, since
\begin{align*}
\begin{array}{cc}
&\left(\dfrac{k_\beta+1}{\,r+k_\beta+2\,}\right)^{m}
\;\leqslant\;
{\displaystyle\prod_{i=1}^{m}\frac{k_\beta+i}{r+k_\beta+1+i}}
\;\leqslant\;
\left(\dfrac{k_\beta+m}{\,r+k_\beta+1+m\,}\right)^{m} \\\\
\iff&\left(1-\dfrac{1+1/r}{1+t_r+2/r}\right)^m
\;\leqslant\;
1-\alpha
\;\leqslant\;
\left(1-\dfrac{1+1/r}{1+t_r+(m+1)/r}\right)^m\\\\
\iff&\dfrac{(1-\alpha)^{1/m}}{1-(1-\alpha)^{1/m}}+o(1)\,\leqslant\, t_r\,\leqslant\,\dfrac{(1-\alpha)^{1/m}}{1-(1-\alpha)^{1/m}}+o(1)\,,
\end{array}
\end{align*} 
for $t_r:=k_\beta(r)/r$. By the \emph{squeeze theorem}, as $r\to\infty$, $\frac{k_\beta(r)}{r}\to \frac{(1-\alpha)^{1/m}}{1-(1-\alpha)^{1/m}}$.


 Next, show the finiteness of $\frac{\frac{k_{\beta}}{r}+1}{\frac{k_{C}}{r}+1}$. First show that if $\phi^*$ is the infimum of Theorem~\ref{thm_gen_sol} for a given $r$ and $k$ then, for any $r'>r$, there exists a unique $k(r')$ such that $k(r')>k$ and the infimum of Theorem~\ref{thm_gen_sol} (using $k(r')$ and $r'$) remains $\phi^*$. Then, argue that $k_{\beta}\leqslant k_{C}$. Together, these imply $\frac {k_{\beta}+r}{k_{C}+r}$ lies between 0 and 1. 
 
 $\phi^*$ is a decreasing function of $r$ when other parameters are fixed; showing this is analogous to our earlier argument showing $\phi^*$ increases with $k$ (see Appendix~\ref{sec_app_phistarconvergence}). Note, the objective function $\phi$ is a continuous function of the $x_i\in \bar{I}_i$, $k$ and $r$: more precisely, it is a continuous function over $\bar{I}_1\times\ldots\times\bar{I}_n\times[0,k_u]\times[0,r_u]$ for any positive $k_u, r_u\in \mathbb R$. Thus, by the \acrfull*{ivt}, $\phi$ attains all of its values between its minimum and maximum over any such cartesian subset. Also note, we know $\phi^*$ is bounded, by the boundedness of $\phi$. All of the foregoing implies that the value of the infimum, $\phi^*$, can be kept fixed as we increase $r$ by increasing $k$ appropriately. In particular, this can be done when $\phi^* = (1-\alpha)$. 
 
For a given $r$, we must have $k_{\beta}\leqslant k_{C}$. Indeed, since the posterior probability of interest is $(1-\alpha)$ using the Beta prior with $r$ and $k_{\beta}$, the infimum of the posterior probability must be smaller than $(1-\alpha)$. So, to make the infimum equal $(1-\alpha)$, we need $k_{C}\geqslant k_{\beta}$, since the infimum is a monotonically increasing function of $k$. Thus, for a given $r$, we have $k_{\beta}\leqslant k_{C}\Longrightarrow 0<\dfrac {\frac{k_{\beta}}{r}+1}{\frac{k_{C}}{r}+1}<1\iff 0\leqslant \liminf\limits_{r\to\infty} \left(\dfrac {\frac{k_{\beta}}{r}+1}{\frac{k_{C}}{r}+1}\right)\leqslant\limsup\limits_{r\to\infty} \left(\dfrac {\frac{k_{\beta}}{r}+1}{\frac{k_{C}}{r}+1}\right)\leqslant1$; i.e. the ratio is asymptotically bounded. 

\newpage
\section{Convergence properties of $y_{*}$ and $y_{**}$ for fixed $\phi^*$}
\label{sec_app_conv_ystar_ystarstar}
The convergence properties of $y_{*}$ and $y_{**}$ are characterized by the following lemma. Fix $m>0$ and a target level $\phi^*\in(0,1)$. For each $r>0$, choose $k=k(r)> 0$
so that the infimum of the objective equals $\phi^*$. Define, for $x\in[0,1]\setminus\{a_r\}$,
\[
h_{r,k}(x):=(1-x)^m\,\frac{r-x(r+m+k)}{\,r-x(r+k)\,},\qquad
a_r:=\frac{r}{r+k},\qquad b_r:=\frac{r}{r+m+k}.
\]
Then $0<b_r<a_r<1$, $h_{r,k}$ is strictly decreasing on $[0,b_r]$ and on $(a_r,1]$, and strictly negative on $(b_r,a_r)$, with
$h_{r,k}(0)=1$, $h_{r,k}(b_r)=h_{r,k}(1)=0$, and $h_{r,k}(x)\to -\infty$ as $x\uparrow a_r$,
$h_{r,k}(x)\to +\infty$ as $x\downarrow a_r$. Consequently, for each $(r,k)$, there are unique $y_{**}(r,k)\in(0,b_r)$, $y_{*}(r,k)\in(a_r,1)$ such that $h_{r,k}\bigl(y_{**}(r,k)\bigr)=h_{r,k}\bigl(y_{*}(r,k)\bigr)=\phi^*$.

\begin{lemma}[convergence of $y_\ast$ and $y_{\ast\ast}$ for fixed $\phi^*$]\label{lem:subseq-dichotomy}
Let $x_*:=1-(\phi^*)^{1/m}\in(0,1)$, and
let $(r_j)_{j\geqslant 1}$ be any increasing sequence with $r_j\to\infty$ along which
$a_{r_j}\to\bar x_\lambda\in[0,1]$. Then, along this sequence,
\[
\boxed{\;
\begin{aligned}
&\text{if }~\bar x_\lambda < x_*,~&& y_{**}(r_j,k(r_j))\longrightarrow \bar x_\lambda,\quad y_{*}(r_j,k(r_j))\longrightarrow x_*,\\
&\text{if }~\bar x_\lambda > x_*,~&& y_{**}(r_j,k(r_j))\longrightarrow x_*,\quad y_{*}(r_j,k(r_j))\longrightarrow \bar x_\lambda,\\
&\text{if }~\bar x_\lambda = x_*,~&& y_{**}(r_j,k(r_j)),\,y_{*}(r_j,k(r_j))\longrightarrow x_*.
\end{aligned}}
\]
This includes $\bar x_\lambda=0$ (i.e.\ $k(r)/r\to\infty$) and
$\bar x_\lambda=1$ (i.e.\ $k(r)/r\to 0$)\footnote{$x_*=\bar{x}_\lambda=1$ is degenerate and unlikely in practice, since this requires $\phi^*=0$.}.
\end{lemma}

\begin{proof} 
\emph{Step 0 (preliminaries).}
Note, $0<b_r<a_r<1$ and $r-x(r+k)=(r+k)(a_r-x)=r(1+k/r)\,(a_r-x)=r(1+t_r)\,(a_r-x)$, where $t_r:=\frac{k}{r}> 0.$
Hence, $\bigl|r-x(r+k)\bigr|=r(1+t_r)\,\bigl|x-a_r\bigr|$. Moreover, $b_r\xrightarrow[r\to\infty]{}\bar x_\lambda$
whenever $a_r\xrightarrow[r\to\infty]{}\bar x_\lambda$.

\emph{Step 1 (uniform convergence away from $\bar{x}_\lambda$).}
Away from $\bar{x}_\lambda$, $h_{r,k}$ converges uniformly to $(1-x)^m$. To show this on the r.h.s. of $\bar{x}_\lambda$ (the l.h.s. is analogous), fix $\delta\in(0,1-\bar{x}_\lambda)$ and $\eta\in\bigl(0,\,1-\bar x_\lambda-\delta\bigr)$, then define the compact set
$K_{\delta,\eta}^+:=[\,\bar x_\lambda+\delta,\,1-\eta\,]\subset(0,1)$.
Since $a_{r_j}\to\bar x_\lambda$, there exists $J_1$ such that, for all $j\geqslant J_1$, $a_{r_j}\leqslant \bar x_\lambda+\delta/2$, whence for every $x\in K_{\delta,\eta}^+$, $x$ is more than a distance $\delta/2$ away from $a_{r_j}$; that is, $|x-a_{r_j}|\ \geqslant\delta/2$. Therefore, $\bigl|r_j-x(r_j+k(r_j))\bigr|
\geqslant(\delta/2)\,r_j(1+t_{r_j})$. Using this, we can uniformly bound a part of the $h_{r,k}$ expression over $K_{\delta,\eta}^+$ as $r_j\rightarrow\infty$. Indeed, 
\[
\left|\frac{r_j-x(r_j+m+k(r_j))}{\,r_j-x(r_j+k(r_j))\,}-1
\right|=\left|\frac{mx}{\,r_j(1+t_{r_j})(a_{r_j}-x)\,}\right|\leqslant\left|\frac{2m(1-\eta)}{r_j(1+t_{r_j})\delta}\right|,
\]
since $x\leqslant 1-\eta$ and $|x-a_r|\geqslant\delta/2$. So, the rational expression on the far left tends to $1$ uniformly
on $K_{\delta,\eta}^+$, and $h_{r_j,k(r_j)}(x)\ \to\ (1-x)^m$ uniformly on $K_{\delta,\eta}^+\,$.

\emph{Step 2 ($y_{*}\rightarrow x_*$ when $\bar x_\lambda<x_*$).}
Choose $\delta\in\bigl(0,\tfrac12(x_*-\bar x_\lambda)\bigr)$ and $\eta\in\bigl(0,\,\tfrac12(1-x_*)\bigr)$,
which ensure $x_*\in K_{\delta,\eta}^+\,$. By Step~1 and the strict decrease of $(1-x)^m$ on $(0,1)$, there exists
$\varepsilon>0$ and $J_2$ such that for all $j\geqslant J_2$,
\[
h_{r_j,k(r_j)}(x_*-\varepsilon)>(1-(x_*-\varepsilon))^m>\phi^*>h_{r_j,k(r_j)}(x_*+\varepsilon)\,.
\]
Since $h_{r_j,k(r_j)}$ is strictly decreasing on $(a_{r_j},1]$, 
$y_*(r_j,k(r_j))\in(x_*-\varepsilon,x_*+\varepsilon)$; hence, $y_*(r_j,k(r_j))\to x_*$.

\emph{Step 3 ($y_{**}\rightarrow\bar{x}_\lambda$ when $\bar x_\lambda<x_*$).}
Fix $\varepsilon\in(0,\tfrac12\bar{x}_\lambda)$ and set $y:=\bar{x}_\lambda-\varepsilon$. Since $b_{r_j}\to\bar x_\lambda$,
we have $y<b_{r_j}$ for all large $j$. By an argument analogous to that used in Step 1, 
$h_{r_j,k(r_j)}(y)\to (1-y)^m$ on any compact subset of $[0,\bar{x}_\lambda]$ containing $y$. Because $y<\bar x_\lambda< x_*$, we have $(1-y)^m>(1-\bar x_\lambda)^m> (1-x_*)^m=\phi^*$. So, for large $j$, we have $h_{r_j,k(r_j)}(y)>\phi^*$. Since $h_{r_j,k(r_j)}$ is strictly decreasing on $[0,b_{r_j}]$
and equals $\phi^*$ at $y_{**}(r_j,k(r_j))$, it follows that $y_{**}(r_j,k(r_j))>y\geqslant \bar x_\lambda-\varepsilon$.
Together with $y_{**}(r_j,k(r_j))<b_{r_j}$ and $b_{r_j}\to\bar x_\lambda$, the \emph{squeeze theorem} yields
$$\bar{x}_\lambda=\lim_{r_j\to\infty}b_{r_j}\geqslant\lim\sup_{r_j\to\infty}y_{**}(r_j,k(r_j))\geqslant\lim\inf_{r_j\to\infty}y_{**}(r_j,k(r_j))\geqslant\bar x_\lambda-\varepsilon\,.$$

\emph{Step 4 (the case $\bar x_\lambda>x_*$).}
An almost symmetrical argument to Steps~1--3 proves the lemma in this case.

\emph{Step 5 (the degenerate pinch $\bar x_\lambda=x_*$).}
For all sufficiently small $\eta>0$, both intervals $[x_*-\eta,\,x_*)$ and $(x_*,\,x_*+\eta]$ lie,
for $j$ large, strictly inside the regions where the uniform convergence in Step~1 (or its left analogue)
applies, and the bracketing in Step~2 holds. That is, eventually, every pair of unique roots is trapped
in $(x_*-\eta,x_*+\eta)$, forcing $y_{**}(r_j,k(r_j))\to x_*\longleftarrow y_{*}(r_j,k(r_j))$.
\end{proof}

\newpage
\section{Proof of the convergence of Corollary~\ref{cor_special_case_no_failure_solution} to the solutions in \cite{strigini_software_2013}}
Let $\phi$ denote the objective function in Corollary~\ref{cor_special_case_no_failure_solution}. View $\phi$ as a function of $y\in I_1$, \begin{align*}
  \phi(y)&=\frac{(1-y)^m(1-y)^kp_1}{(1-y)^kp_1+\sum_{i=2}^{n}(1-x_i)^kp_i}+\frac{\sum_{i=2}^{n}(1-x_i)^m(1-x_i)^kp_i}{(1-y)^kp_1+\sum_{i=2}^{n}(1-x_i)^kp_i}     
   \end{align*} The first rational function on the r.h.s. grows, while the second rational function decays,  as $y\to 0$. 
   Since the $x_i$s and $y$ are strictly ordered, $(1-y)^m>(1-x_2)^m>\ldots>(1-x_n)^m$ holds, thus $\phi(y)$ is monotonically increasing as $y\to 0$. This means, $\phi(y_{11})<\phi(y_{12})$ for all $y_{11}>y_{12}$, where $y_{11}$ and $y_{12}$ are possible upper endpoints of the first interval. Thus, $\underset{x_2\in I_2,\ldots x_n\in I_{n}}{\inf}  \phi(y_{11})\leqslant\underset{x_2\in I_2,\ldots x_n\in I_{n}}{\inf} \phi(y_{12})$ for all $y_{11}>y_{12}$: i.e. $\phi^*$ also monotonically increases, as the upper endpoint of the first interval decreases. With $\phi^*\leqslant1$,  the \emph{completeness of the reals} implies $\lim\limits_{y\to0}\phi^*$ exists.  
   So, letting $y_1\to0$ in Corollary~\ref{cor_special_case_no_failure_solution}, \begin{align}
\lim\limits_{y_1\to 0}\phi^*
=&\dfrac{p_1+\sum\limits_{i=2}^{j-1}(1-y_{i}^*)^{m+k}p_{i}+(1-\lim\limits_{y_1\to 0}y_*)^{m+k}p_{j}+\sum\limits_{i=j+1}^{n}(1-y_{i-1}^*)^{m+k}p_{i}}{p_1+\sum\limits_{i=2}^{j-1}(1-y_{i}^*)^{k}p_{i}+(1-\lim\limits_{y_1\to 0}y_*)^{k}p_{j}+\sum\limits_{i=j+1}^{n}(1-y_{i-1}^*)^{k}p_{i}} \label{eq_limit}
\end{align}
where $y_i^*=\lim\limits_{y_1\to 0}y_i$ for $i\in\{2,\ldots,{j-1},{j+1},\ldots,n\}$.
   
Now, $y_*$ uniquely satisfies $\phi^*=\left(\frac{m+k}{k}\right)(1-y_*)^{m}$ for $y_{j-1}\leqslant y_*\leqslant y_j$, so $\lim\limits_{y_1\to 0}\phi^*=\left(\frac{m+k}{k}\right)(1-\lim\limits_{y_1\to 0}y_*)^{m}$. Thus $\lim\limits_{y_1\to 0}y_*$ exists, because $\lim\limits_{y_1\to 0}\phi^*$ exists. 
This also implies the existence of all $y_i^*$. Consequently, \eqref{eq_limit} exists and is unique. In particular, for $n=2$, \eqref{eq_limit} yields $\lim\limits_{y_1\to0}\phi^*
 =(p_1+(1-\lim\limits_{y_1\to0}y_*)^{m+k}p_2)/(p_1+(1-\lim\limits_{y_1\to0}y_*)^{k}p_2)$. We conclude, as $y_1\to 0$ (i.e. in the limit of a non-zero probability of perfection), Corollary~\ref{cor_special_case_no_failure_solution} converges to the solutions obtained by \cite{strigini_software_2013}. 

\noindent \textbf{Remark:} 
For $n\leqslant3$ and $r>0$, $\phi^*$ in Theorem~\ref{thm_gen_sol} tends to zero as $y_1\to0$. Assume $n=3$ and $r>0$. The objective function in Theorem~\ref{thm_gen_sol}, denoted $\phi$, becomes $\phi=(\sum_{i=2}^3x_i^r(1-x_i)^{m+k}p_i)/(\sum_{i=2}^3x_i^r(1-x_i)^kp_i)$, where $0<y_1<x_2\leqslant y_2<x_3\leqslant1$. In the limit as $y_1\to0$, $\phi$ becomes $(1-x_3)^m$ where $0<x_3\leqslant1$. Since $\phi\to0$ as $x_3\to1$, we have $\phi^*=0$. 

Similarly, for $n=2$, $\phi=(1-x_2)^m$. Since $\phi\to0$ as $x_2\to1$, then $\phi^*=0$.


\end{document}